\documentclass[a4paper,11pt]{article}
\usepackage{amsthm}

\usepackage{enumerate}
\usepackage{algorithm}% http://ctan.org/pkg/algorithms
\usepackage{algpseudocode}% http://ctan.org/pkg/algorithmicx
\usepackage[a4paper,margin=1in]{geometry}

\usepackage{amsmath}
\usepackage[capitalise]{cleveref}

\crefname{equation}{}{}

\usepackage[english]{babel}
\usepackage{graphicx}

\usepackage{amssymb}

\newtheorem{thm}{Theorem}[section]
\newtheorem{lem}[thm]{Lemma}
\newtheorem{obs}[thm]{Observation}

\theoremstyle{definition}
\newtheorem{defn}[thm]{Definition}

\newtheorem{exmp}[thm]{Example}

\newcommand{\T}{\mathcal{T}}
\newcommand{\Leav}{\mathcal{L}}
\newcommand{\N}{\mathcal{N}}

\usepackage{subfiles}
\usepackage{placeins}
\usepackage{newicktree}
\usepackage{qtree}
\usepackage{tikz}
\usepackage{xspace}
\usepackage[edges]{forest}
\usepackage{subcaption}

\usepackage{todonotes}
\usepackage{soul}
\usepackage{epsfig}
\usepackage{nomencl}

\usepackage{mathtools}
\Crefname{line}{line}{lines}

\newcommand\problem[3]{%
\vspace{0.5cm}
\par\noindent {\bfseries  #1}\\
{\bfseries Instance}: #2\\
{\bfseries Question}: #3\par
\vspace{0.5cm}
}
\newcommand{\doctype}{paper\xspace}
\newcommand{\markj}[1]{#1}
\newcommand{\leo}[1]{#1}
\newcommand{\sander}[1]{#1}

\title{New FPT algorithms for finding the temporal hybridization number for sets of phylogenetic trees\thanks{Leo van Iersel and Mark Jones were partly supported by the Netherlands Organization for Scientific Research (NWO), Vidi grant 639.072.602 and Mark Jones also by the gravitation grant NETWORKS.}}
\author{Sander Borst\footnote{Centrum Wiskunde \& Informatica (CWI), P.O. Box 94079, 1090 GB Amsterdam, The Netherlands, \texttt{Sander.Borst@cwi.nl, markelliotlloyd@gmail.com}} \and Leo van Iersel\footnote{Delft Institute of Applied Mathematics, Delft University of Technology, Van Mourik Broekmanweg 6, 2628 XE, Delft, The Netherlands, \texttt{L.J.J.vanIersel@tudelft.nl}.}\and Mark Jones\footnotemark[2] \and Steven Kelk\footnote{Department of Data Science and Knowledge Engineering (DKE)
Maastricht University
P.O. Box 616
6200 MD Maastricht 
The Netherlands, \texttt{steven.kelk@maastrichtuniversity.nl}.}}

\begin{document}
%\sloppy
\maketitle  

\begin{abstract}
	We study the problem of finding a temporal hybridization network for a set of phylogenetic trees that minimizes the number of reticulations.
	First, we introduce an FPT algorithm for this problem on an arbitrary set of $m$ binary trees with $n$ leaves each with a running time of $O(5^k\cdot n\cdot m)$, where $k$ is the minimum temporal hybridization number.
	We also present the concept of \emph{temporal distance}, which is a measure for how close a tree-child network is to being temporal.
	Then we introduce an algorithm for computing a tree-child network with temporal distance at most $d$ and at most $k$ reticulations in $O((8k)^d5^ k\cdot n\cdot m)$ time.
	Lastly, we introduce a \leo{$O(6^kk!\cdot k\cdot n^2)$} time algorithm for computing a minimum temporal hybridization network for a set of two nonbinary trees.
	We also provide an implementation of all algorithms and an experimental analysis on their performance.
\end{abstract}
%\newpage
%\tableofcontents

%\newpage

% --------- INTRODUCTION ----------

\section{Introduction}
Phylogenetics is the study of the evolutionary history of biological species.
Traditionally such a history is represented by a phylogenetic tree.
However, hybridization and horizontal gene transfer, both so-called \emph{reticulation} events, can lead to multiple seemingly conflicting trees representing the evolution of different parts of the genome \cite{mallet_how_2016, soucy_horizontal_2015}.
Directed acyclic networks can be used to combine these trees into a more complete representation of the history \cite{bapteste_networks_2013}.
Reticulations are represented by vertices with in-degree greater than one.

Therefore, an important problem is how to construct such a network based on a set of input trees that are known to represent the evolutionary history for different parts of the genome.
The network should display all of these input trees.
In general there are many solutions to this problem, but in accordance with the parsimony principle we are especially interested in the most simple solutions to the problem.
These are the solutions with a minimal number of reticulations.
Finding a network for which the number of reticulations, also called the \emph{hybridization number}, is minimal now  becomes an optimization problem.
This problem is NP-complete, even for only two binary input trees \cite{bordewich_computing_2007-2}.
The problem is fixed parameter tractable for an arbitrary set of non-binary input trees if either the number of trees or the out-degree in the trees is bounded by a constant \cite{van_iersel_kernelizations_2016-1}.
For a set of two binary input trees an FPT algorithm with a reasonable running time exists \cite{bordewich_computing_2007}.
For more than two input trees theoretical FPT algorithms and practical heuristic algorithms exist, but no FPT algorithm with a reasonable running time is known.
That is why we are interested in slightly modifying the problem to make it easier to solve.

One way to do this is by restricting the solution space to the class of tree-child networks, in which each non-leaf vertex has at least one outgoing arc that does not enter a reticulation \cite{cardona_comparison_2007}.
The minimum hybridization number over all tree-child networks that display the input trees is called the \emph{tree-child hybridization number}.
These networks can be characterized by so-called cherry picking sequences \cite{linz_attaching_2019}.
This characterization can be used to create a fixed parameter tractable algorithm for this restricted version of the problem for any number of binary input trees with time complexity $O((8k)^k\cdot poly(n, m))$ where $k$ is the tree-child hybridization number, $n$ is the size of leaves and $m$ is the number of input trees \cite{van_iersel_practical_2019}.

The solution space can be reduced even further \cite{humphries_cherry_2013}, leading to the problem of finding the \emph{temporal hybridization number}.
The extra constraints enforce that each species can be placed at a certain point in time such that evolution events take a positive amount of time and that reticulation events can only happen between species that live at the same time.
For the problem of computing the temporal hybridization number a cherry picking characterization exists too and it can be used to develop a fixed parameter tractable algorithm for problems with two binary input trees with time complexity $O((7k)^k\cdot poly(n, m))$ where $k$ is the temporal hybridization number, $n$ is the number of leaves and $m$ is the number of input trees \cite{humphries_cherry_2013}.
In this \doctype we introduce a faster algorithm for solving this problem in $O(5^k \cdot n \cdot m)$ time using the cherry picking characterization.
Moreover, this algorithm works for any number of binary input trees.

A disadvantage of the temporal restrictions is that in some cases no solution satisfying the restrictions exists.
In fact determining whether such a solution exists is a NP-hard problem \cite{humphries_complexity_2013}\cite{docker_deciding_2019}.
Because of this our algorithm will not find a solution network for all problem instances.
However we show that it is possible to find a network with a minimum number of non-temporal arcs, thereby finding a network that is `as temporal as possible'.
For that reason we also introduce an algorithm that also works for non-temporal instances.
This algorithm is a combination of the algorithm for tree-child networks and the one for temporal networks introduced here.

In practical data sets, the trees for parts of the genome are often non-binary.
This can be either due to simultaneous divergence events or, more commonly, due to uncertainty in the order of divergence events \cite{linz_hybridization_2009}.
This means that many real-world datasets contain non-binary trees, so it is very useful to have algorithms that allow for non-binary input trees.
While the general hybridization number problem is known to be FPT when either the number of trees or the out-degree of the trees is bounded by a constant \cite{van_iersel_kernelizations_2016-1}, an FPT algorithm with a reasonable running time ($O(6^kk!
    \cdot poly(n))$) is only known for an input of two trees \cite{piovesan_simple_2013}.
Until recently no such algorithm was known for the temporal hybridization number problem however.
In this \doctype the first FPT algorithm  for constructing optimal temporal networks based on two non-binary input trees with running time \leo{$O(6^kk!\cdot k\cdot n^2)$} is
%$O(4^kk!)$
introduced.

\leo{We implemented and tested all new algorithms~\cite{sjb_implementation}.}

The structure of the paper is as follows.
First we introduce some common theory and notation in \cref{sec:preliminaries}.
In \cref{sec:algorithm} we present a new algorithm for the temporal hybridization number of binary trees, prove its correctness and analyse the running time.
In \cref{sec:non_temporal} we combine the algorithm from \cref{sec:algorithm}  with the algorithm from \cite{van_iersel_practical_2019} to obtain an algorithm for constructing tree-child networks with a minimum number of non-temporal arcs.
In \cref{sec:non_binary_trees} we present the algorithm for the temporal hybridization number for two non-binary trees.
In \cref{sec:implementation} we conduct an experimental analysis of the \leo{algorithms}.

% -------- PRELIMINARIES ----------

\section{Preliminaries}
\label{sec:preliminaries}
\subsection{Trees}
A \emph{rooted binary phylogenetic $X$-tree} $\mathcal{T}$ is a rooted binary tree for which the leaf set is \leo{equal to} $X$ with $|X|=n$.
Because we will mostly use rooted binary phylogenetic trees in this \doctype we will just refer to them as \emph{trees}.
Only in \cref{sec:non_binary_trees} trees that are not necessarily binary are mentioned, but we will explicitly call them non-binary trees.

Each of the leaves of a tree \leo{is an element} of $X$.
We will also refer to the set of \leo{leaves} in $\T$ as $\Leav(\T)$.
For a tree $\mathcal{T}$ and a set of leaves $A$ with the notation $\mathcal{T}\setminus A$ we refer to the tree obtained by removing all leaves \leo{that are in} $A$ from $\T$ and repeatedly contracting all vertices with both in- and out-degree one.
Observe that $\left(\mathcal{T}\setminus \{x\}\right)\setminus \{y\} = \mathcal{T}\setminus \{x,y\} = \left(\mathcal{T}\setminus \{y\}\right)\setminus \{x\}$.
We will often use $T$ to refer to a set of $m$ trees $\mathcal{T}_1,\ldots, \mathcal{T}_m$.
We will write $T\setminus A$ for $\{\T_1\setminus A,\ldots, \T_m\setminus A \}$ and $\Leav(T)=\cup_{i=1}^m\Leav(\T_i)$.

\subsection{Temporal networks}

A \emph{network} \leo{on~$X$} is a rooted acyclic directed graph satisfying:
\begin{enumerate}
	\item The root $\rho$ has in-degree $0$ and an out-degree not equal to $1$.
	\item The \emph{leaves} are the nodes with out-degree zero. \leo{The set of leaves is~$X$.}
	\item The remaining vertices are \emph{tree vertices} or \emph{hybridization vertices}
	      \begin{enumerate}
		      \item A tree vertex has in-degree $1$ and out-degree at least $2$.
		      \item A hybridization vertex (also called \emph{reticulation}) has out-degree $1$ and in-degree at least $2$.
	      \end{enumerate}
\end{enumerate}

We will call the arcs ending in a hybridization vertex \emph{hybridization arcs}.
All other arcs are \emph{tree arcs}.
A network is a \emph{tree-child} network if every tree vertex has at least one outgoing tree arc.

We say that a network $\mathcal{N}$ on $X$ displays a set of trees $T$ on $X'$ with $X'\subseteq X$ if every tree in $T$ can be obtained by removing edges and vertices and contracting vertices with both in-degree $1$ and out-degree $1$.
For a set of leaves $A$ we define $\N\setminus A$ to be the network obtained from $\N$ by removing all leaves in $A$ and afterwards removing all nodes with out-degree zero and contracting all nodes with both in- and out-degree one.

\begin{figure}
	\centering
	\begin{subfigure}[b]{.25\textwidth}
		\includegraphics{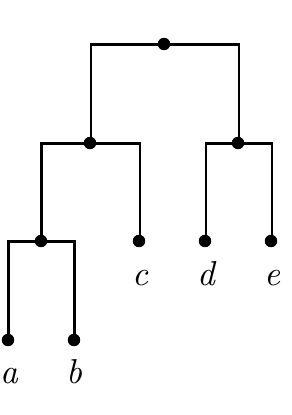}
		\caption{
			\label{subfig:first_tree}}
	\end{subfigure}
	\begin{subfigure}[b]{.25\textwidth}
		\includegraphics{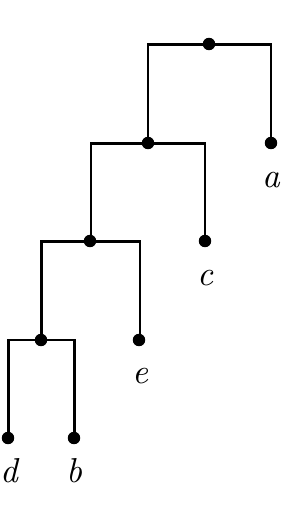}
		\caption{
			\label{subfig:second_tree}
		}
	\end{subfigure}
	\begin{subfigure}[b]{.45\textwidth}
		\includegraphics{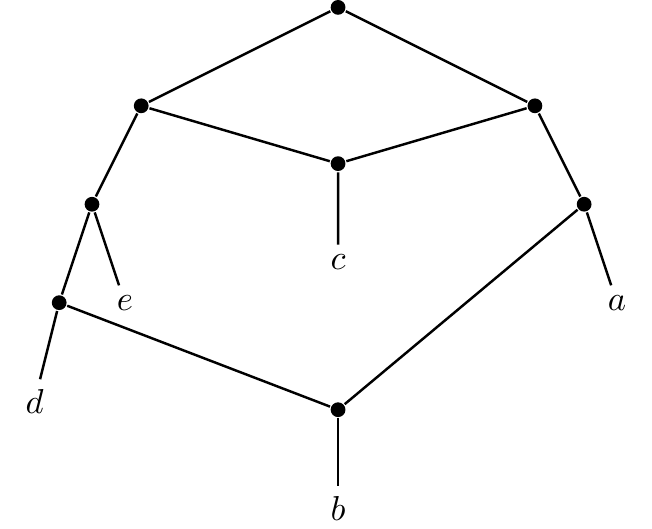}
		\caption{\label{subfig:network}}
	\end{subfigure}
	\caption{The binary trees in (a) and (b) are both displayed by the network in (c). }
	\label{fig:simple_example}

\end{figure}

For a tree-child network $\mathcal{N}$, the \emph{hybridization number} $h_t(\mathcal{N})$ is defined as
\begin{align*}
	r(\mathcal{N})=\sum_{v\neq \rho}(d^-(v)-1)\text{.
	}
\end{align*}
where $d^-(v)$ is the in-degree of a vertex $v$ and $\rho$ is the root of $\mathcal{N}$.

A tree-child network $\mathcal{N}$ with set of vertices $V$ is \emph{temporal} if there exists a map $t:V\to \mathbb{R}^+$, called a temporal labelling, such that for all $u,v\in V$ we have $t(u)=t(v)$ when $(u,v)$ is a hybridization arc and $t(u)<t(v)$ when $(u,v)$ is a tree arc.
In \cref{fig:temporal_and_non_temporal_network} both a temporal and a non-temporal network are shown.

\begin{figure}
	\begin{subfigure}[b]{.45\textwidth}
		\includegraphics[scale=.9]{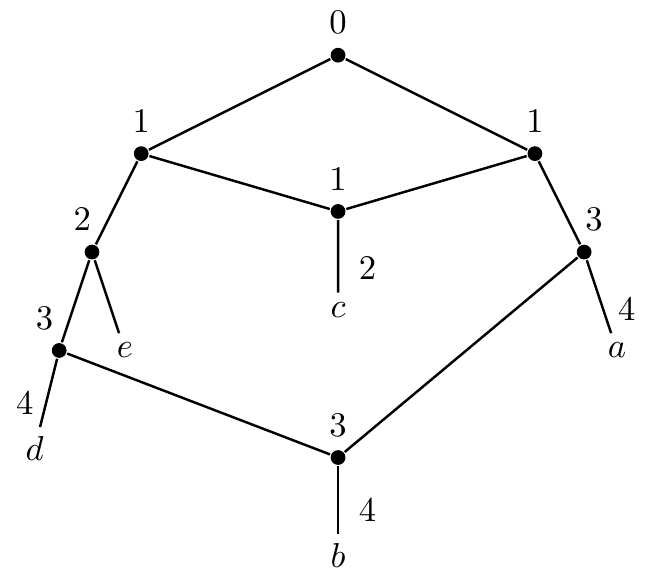}
		\caption{
			A temporal labeling is shown in the network above, asserting that the network is temporal.
			\label{subfig:network_temporal}
		}
	\end{subfigure}
	\hfill
	\begin{subfigure}[b]{.45\textwidth}
		\includegraphics[scale=.9]{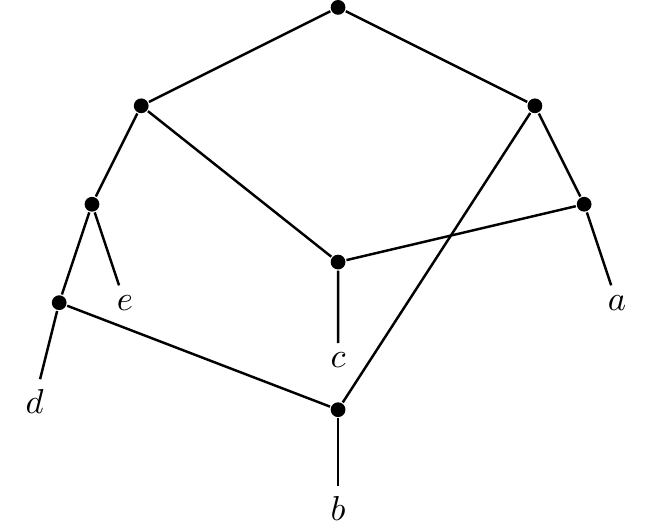}
		\caption{
			No temporal labeling exists for this network.
			Therefore the network is not temporal.
			\label{subfig:network_non_temporal}}
	\end{subfigure}
	\caption{\label{fig:temporal_and_non_temporal_network}}
\end{figure}

For a set of trees $T$ we define the minimum temporal-hybridization number as
\begin{align*}
	h_t(T)=\min \{r(\mathcal{N}):\mathcal{N}\text{ is a temporal network that displays }T \}
\end{align*}
This definition leads to the following decision problem.

\problem{Temporal hybridization}{A set of trees $T$ and an integer $k$}{Is $h_t(T)\leq k$?}

Note that there are sets of trees such that no temporal network exists that displays them.
In 		\cref{fig:non_temporal_trees} an example is given.
For such a set $T$ we have $h_t(T)=\infty$.
\begin{figure}
	\centering
	\includegraphics{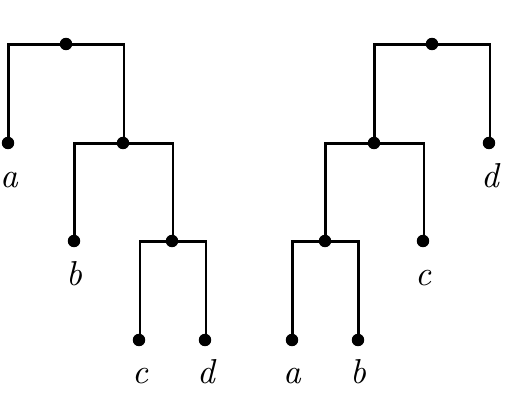}
	\caption{
		No temporal network that displays these trees exists.
		\label{fig:non_temporal_trees}
	}

\end{figure}

\subsection{Cherry picking sequences}

Temporal networks can now be characterized by so-called cherry-picking sequences \cite{humphries_cherry_2013}.
A \emph{cherry} is a set of children of a tree vertex that only has leaves as children.
So for binary trees a cherry is a pair of leaves.
We will write $(a,b)\in \mathcal{T}$ if $\{a,b\}$ is a \emph{cherry} of $\mathcal{T}$ and $(a,b)\in T$ if there is a $\mathcal{T}\in T$ with $(a,b)\in \mathcal{T}$.
First we introduce some notation to make it easier to speak about cherries.

\begin{defn}
	\label{def:ht}
	For a set of binary trees $T$ on the same taxa define $H(T)$ to be the set of leaves that is in a cherry in every tree.
\end{defn}

If two leaves are in a cherry together we call them \emph{neighbors}.
We also introduce notation to speak about the \emph{neighbors} of a given leaf: 

\begin{defn}
	Define $N_\mathcal{T}(x) =\{y\in\mathcal{X}:  (y,x)\in \mathcal{T} \}$.
	For a set of trees $T$ define $N_T(x)=\cup_{\mathcal{T}\in T}N_\mathcal{T}(x)$.
\end{defn}

\begin{defn}
	For a set of binary trees $T$ containing a leaf $x$ define $w_T(x)=|N_T(x)| -1$.
	We will also call this the \emph{weight} of $x$ in $T$.
	\label{def:weight}
\end{defn}

Using this theory, we can now give the definition of cherry picking sequences.

\begin{defn}
	A sequence of leaves $s=(s_1,s_2,\ldots, s_n)$ is a \emph{cherry picking sequence} (CPS) for a set of binary trees $T$ on the same set of taxa if it contains all leaves of $T$ exactly once and if for all $i\in [n-1]$ we have $s_i \in H(T\setminus \{s_1, \ldots, s_{i-1} \})$.
	The weight $w_T(s_1,\ldots s_n)$ of the sequence is defined as $w_T(s)=\sum_{i=1}^{n-1}w_{T\setminus \{s_1, \ldots, s_{i-1} \}} (s_{i})$.
	\label{def:cps}
\end{defn}
\begin{exmp}
	\label{exmp:simple_example}
	For the two trees in \cref{fig:simple_example}, $\leo{(\boldsymbol{b},e,\boldsymbol{c},d,a)}$ is a minimum weight cherry-picking sequence of weight $2$.
	Leaves $b$ and $c$ (indicated in bold) have weight $1$ and the rest of the leaves have weight $0$ in the sequence.
\end{exmp}
For a cherry picking sequence $s$ with $s_i=x$ we say that $x$ is \emph{picked} in $s$ at index $i$.
\begin{thm}[{\cite[Theorem 1, Theorem 2]{humphries_cherry_2013}}]
	\label{lem:exists_cherry_sequence}
	Let $T$ be a set of trees on $\mathcal{X}$.
	There exists a temporal network $\mathcal{N}$ that displays $T$ with $h_t(\mathcal{N})=k$ if and only if there exists a cherry-picking sequence $s$ for $T$ with $w_T(s)=k$.
\end{thm}
This has been proven in \cite[Theorem 1, Theorem 2]{humphries_cherry_2013}.
The proof works by constructing a cherry picking sequence from a temporal network and vice versa.
Here, we only repeat the construction to aid the reader, and refer to \cite{humphries_cherry_2013} for the proof of correctness.

The construction of cherry picking sequence $s$ from a temporal network $\N$ with temporal labeling $t$ works in the following way: For $i=1$ choose $s_i$ to be a leaf $x$ of $\N$ such that $t(p_x)$ is maximal where $p_x$ is the parent of $x$ in $\N$.
Then increase $i$ by one and again choose $s_i$ to be a leaf $x$ of $\N\setminus \{s_1, \ldots, s_{i-1} \}$ that maximizes $t(p_x)$ where $p_x$ is the parent of $x$ in $\N\setminus \{s_1, \ldots, s_{i-1} \}$.
In \cite[Theorem 1, Theorem 2]{humphries_cherry_2013} it is shown that now $s$ is a cherry picking sequence with $w_T(s)=r(\N)$.

The construction of a temporal network $\N$ from a cherry picking $s$ is somewhat more technical: for cherry picking sequence $s_1,\ldots, s_t$, define $\N_{n}$ to be the tree, only consisting of a root and \leo{leaf~$s_n$}
%the only leaf of $T\setminus \{s_1, \ldots, s_n \}$.
Now obtain $\N_{i}$ from $\N_{i+1}$ by adding node $s_i$ and a new node $p_{s_i}$, adding edge $(p_{s_i},s_i)$ subdividing $(p_x,x)$ for every $x\in N_{T\setminus \{s_1, \ldots, s_{i-1} \}}(s_i)$ with node $q_x$ and adding an edge $(q_x,p_{s_i})$ and finally suppressing all nodes with in- and out-degree one.
Then $\N=\N_1$ displays $T$ and $r(\N)=w_T(s)$.

The theorem implies that the weight of a minimum weight CPS is equal to the temporal hybridization number of  the trees.
Because finding an optimal temporal reticulation network for a set of trees is an NP-hard problem \cite{humphries_complexity_2013}, this implies that finding a minimum weight CPS is an NP-hard problem.

\begin{defn}
	We call two sets of trees $T$ and $T'$ \emph{equivalent} if a bijection from $\mathcal{L}(T)$ to $\mathcal{L}(T')$ exists that transforms $T$ into $T'$.
	We call them equivalent because have the same structure and consequently the same (temporal-) hybridization \leo{number}, however the biological interpretation can be different.
	We will write this as $T\simeq T'$.
	\label{defn:equivalent}
\end{defn}

% ------------- THE ALGORITHM --------------

\section{Algorithm for constructing temporal networks from binary trees}
\label{sec:algorithm}
\label{section:foundation_of_v3}

Finding a cherry picking sequence comes down to deciding in which order to pick the leaves.
Our algorithm relies on the observation that this order does not always matter.
Intuitively the observation is that the order of two leaves in a cherry picking sequence only matters if they appear in a cherry together somewhere during the execution of the sequence.
Therefore the algorithm keeps track of the pairs of leaves for which the order of picking matters.
We will make this more precise in the remainder of this section.
The algorithm now works by branching on the choice of which element of a pair to pick first.
These choices are stored in a so-called constraint set.
Each call to the algorithm branches into subcalls with more constraints added to the constraint set.
As soon as it is known that a certain leaf has to be picked before all of its neighbors and is in a cherry in all of the trees, the leaf can be picked.

\begin{defn}
	Let $C\subseteq \mathcal{L}(T) \times \mathcal{L}(T)$.
	We call $C$ a \emph{constraint set} on $T$ if every pair $(a,b)\in C$ is a cherry in $T$.
	A cherry picking sequence $s=(s_1,\ldots, s_k)$ of $T$ \emph{satisfies} $C$ if for all $(a,b)\in C$, we have $s_i=a$ and $(a,b)\in T'$ and $w_{T'}(a)>0$ with $T'=T\setminus \{s_1,\ldots , s_{i-1} \}$ for some $i$.
\end{defn}
Intuitively, a cherry picking sequence satisfies a constraint set if for every pair $(a,b)$ in the set $a$ is picked with positive weight and $(a,b)$ is a cherry just before picking $a$.
This implies that $a$ occurs in the cherry picking sequence before $b$.

\markj{We now prove a series of results about what sets of constraints are valid, which will then be used to guide our algorithm.}

\begin{figure}
	\centering
	\includegraphics[scale = 0.2]{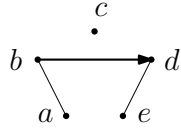}
	\caption{
	\markj{
	An example showing the neighbour relation for the trees in Figure~\ref{fig:simple_example}, together with a constraint $(b,d)$.
    Two elements $x,y \in X$ are depicted as adjacent if $x \in N_T(y)$ i.e. if $x$ and $y$ appear in a cherry together.
    An arc from $x$ to $y$ indicates the presence of a constraint $(x,y)$.
    }
		\label{fig:cherry_adjacency_example}
	}
\end{figure}

\begin{obs}\label{obs:ConstraintThreeCases}
	Let $s$ be a cherry picking sequence for $T$ and $w_T(x) > 0$ and $a,b\in N_T(x)$.
	Then $s$ satisfies one of the following constraint sets: \\$\{(a,x)\}, \{(b,x)\}, \{(x,a),(x,b)\}$. 
	%\todo[color=green!40]{add figures with example}
	\label{lem:branch_in_three}
\end{obs}
\begin{proof}
	Let $i$ be the lowest index such that $s_i \in\{x,a,b\}$.
	If $s_i=x$, then $(x,a)\in T\setminus \{s_1,\ldots, s_{i-1} \}$ and $(x,b)\in T\setminus \{s_1,\ldots, s_{i-1} \}$, so $s$ satisfies $\leo{\{(x,a),(x,b)\}}$.
	If $s_i=a$, then there is a $\T\in T\setminus \{s_1,\ldots, s_{i-1} \}$ with $(x,b)\in \T$, so $(a,x)\notin \T$, which implies that $w_{T\setminus \{s_1,\ldots, s_{i-1} \}}(s_i)>0$, so $s$ satisfies $\{(a,x)\}$.
	Similarly if $s_i=b$ then $s$ satisfies $\{(b,x)\}$.
\end{proof}
\begin{exmp}\label{ex:ConstraintThreeCases}
    The trees in \cref{subfig:first_tree} and \cref{subfig:second_tree} contain the cherries $(a,b)$ and $(d,b)$. So by \cref{lem:branch_in_three} every cherry picking sequence for these trees satisfies one of the constraint sets $\{(a,b)\}, \{(d,b)\}, \{(b,a),(b,d)\}$. For example, $\leo{({\bf b},d,{\bf c},e,a)}$ is a cherry picking sequence of weight $2$ for these trees. This sequence satisfies the constraint set $\{(b,a),(b,d)\}$. \markj{See Figure~\ref{fig:constraint_three_cases}.} %\todo{check the boldface in the CPS}
\end{exmp}

\begin{figure}
	\centering
	\hfill
	\begin{subfigure}{.3\textwidth}
	\includegraphics[scale=.3]{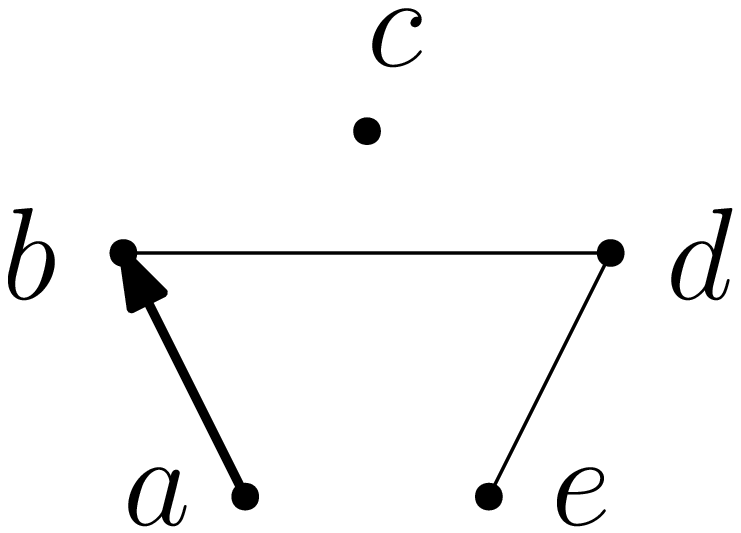}
	\end{subfigure}
	\hfill
	\begin{subfigure}{.3\textwidth}
	\includegraphics[scale=.3]{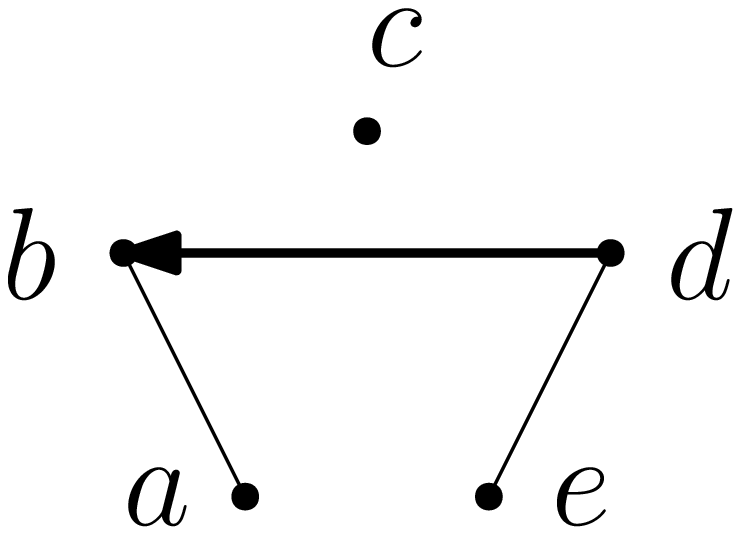}
	\end{subfigure}
	\hfill
	\begin{subfigure}{.3\textwidth}
	\includegraphics[scale=.3]{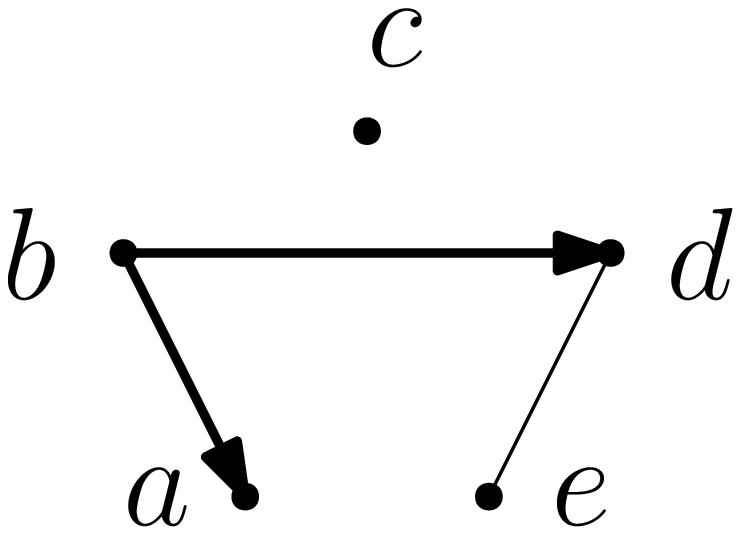}
	\end{subfigure}
	\hfill
	\caption{
	\markj{
    Illustration of Example~\ref{ex:ConstraintThreeCases}, showing the possible constraint sets on $a,b,d$ implied by 
	Observation~\ref{obs:ConstraintThreeCases}.
    }
		\label{fig:constraint_three_cases}
	}
\end{figure}

This observation implies that the problem can be reduced to three subproblems, corresponding to either appending $\{(a,x)\}$, $\{(b,x)\}$ or  $\{(x,a),(x,b)\}$ to $C$.
As we will see, this is used by the algorithm.
It is possible to implement an algorithm using only this rule, but the running time of the algorithm can be improved by using a second rule that branches into only two subproblems when it is applicable.
The rule relies on the following observation. Note that we will write $\pi_i(C)$ for the set obtained by projecting every element of $C$ to the $i$'th coordinate.
\begin{obs}\label{obs:ConstraintTwoCases}
	If $C$ is satisfied by $s$ then for all $x\in \pi_1(C)$ and $y\in N_T(x)$ we have that either $C\cup \{(y,x)\}$ or $C\cup \{ (x,y)\}$ is also satisfied by $s$.
	\label{lem:branch_in_two}
\end{obs}
\begin{proof}
	If $x\in \pi_1(C)$ then $C$ contains a pair $(x, a)$.
	If $a=y$ it is trivial that $s$ satisfies $C\cup \{ (x,y)\}=C$.
	Otherwise \cref{lem:branch_in_three} implies that $s$ satisfies one of the constraint sets $\{(a,x)\}, \{(y,x)\}, \{(x,a),(x,y)\}$.
	Because $s$ satisfies $\{(x,a)\}$, $s$ can not satisfy $\{(a,x)\}$.
	So $s$ will satisfy either $\{(y,x)\}$ or $\{(x,a),(x,y)\}$.
\end{proof}
Using this observation we can let the algorithm branch into two paths by either adding $(x,y)$ or $(y,x)$ to the constraint set $C$ if $x\in\pi_1(C)$. 

\begin{exmp}\label{ex:ConstraintTwoCases}
\leo{Consider again the situation in Example~\ref{ex:ConstraintThreeCases}. Suppose we guess that the solution satisfies the constraint set $\{(d,b)\}$. Then we have~$d\in\pi_1(C)$. Hence, we are in the situation of Observation~\ref{obs:ConstraintTwoCases} and we can conclude that either $(d,e)$ or $(e,d)$ can be added to the constraint set~$C$. See Figure~\ref{fig:constraint_two_cases}.}
\end{exmp}

\begin{figure}
	\centering
	\begin{subfigure}{.4\textwidth}
	\includegraphics[scale=.65]{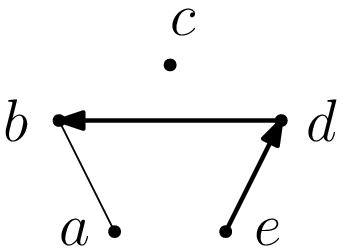}
	\end{subfigure}
	\begin{subfigure}{.4\textwidth}
	\includegraphics[scale=.65]{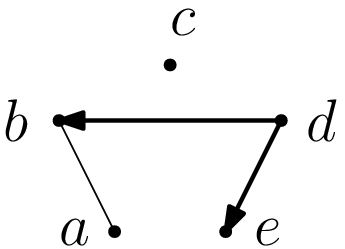}
	\end{subfigure}
	\caption{
	\markj{
    Illustration of Example~\ref{ex:ConstraintTwoCases}, showing the two possible constraints on $d$ and $e$ implied by 
	Observation~\ref{obs:ConstraintTwoCases}, in the case that there already exists a constraint $\leo{(d,b)} \in C$ and thus $d \in \pi_1(C)$.
    }
		\label{fig:constraint_two_cases}
	}
\end{figure}
 %\todo{mark, please check, Fig.\ref{fig:constraint_two_cases} and Ex.~\ref{ex:ConstraintTwoCases}}
 
We define $G(T,C)$ to be the set of cherries for which there is no constraint in $C$, so $G(T,C)=\{(x,y):(x,y) \in T \land (x,y), (y,x) \notin C  \}$.
Observe that $(x,y)\in G(T,C)$ is equivalent with $(y,x) \in G(T,C)$.

\markj{Before proving the next result about constraints, we need the following lemma. This}
states that if we have a set of trees, a leaf that is in a cherry in all of the trees and a corresponding cherry picking sequence then the following holds: for every element in a cherry picking sequence, we can either move it to the front of the sequence without affecting the weight of the sequence or there is a neighbor of this element that occurs earlier in the sequence.
\begin{lem}
	\label{lem:twoconditions}
	%\todo[color=green!40]{add intuitive explanation}
	Let $(s_1,s_2,\ldots)$ be a cherry picking sequence for a set of trees $T$ that satisfies constraint set $C$.
	Let $x\in H(T)$.
	Then at least one of the following statements is true:
	\begin{enumerate}
		[ {(}1{)} ] \item $\exists i: s_i=x $ and $s'=(s_i,s_1,\ldots ,s_{i-1},s_{i+1},\ldots)$ is a cherry picking sequence for $T$ satisfying $C$ and $w(s)=w(s')$.
		\item If $s_i=x$ then $\exists j: s_j \in N_T(x)$ such that $j<i$.
	\end{enumerate}
\end{lem}
\begin{proof}
	Let $r$ be the smallest number such that $s_r\in N_T(x) \cup \{x\}$.
	In case $s_r\neq x$ it follows directly that condition (2) holds for $j=r$. 
	For $s_r=x$ we will prove that condition (1) holds with $i=r$. 
	\markj{The key idea is that, because $s_i$ is not in a cherry with any of $s_1,\ldots, s_{i-1}$, removing $s_i$ first will not have any effect on the cherries involving $s_1,\ldots, s_{i-1}$.}

    %\todo[color=green!40]{Replaced a number of occurences of $\T(s'_1,\ldots, s'_{j-1})$ with  $\T\setminus\{s'_1,\ldots, s'_{j-1}\}$ (or similar)}
	\markj{More formally,} take an arbitrary tree $\T\in T$.
	Now take arbitrary $j,k$ with $s'_j=s_k$.
	Now we claim that for an arbitrary $z$ we have 
	\markj{$(s'_j,z) \in \T\setminus\{s'_1,\ldots, s'_{j-1}\}$}
%	$(s'_j,z) \in \T(s'_1,\ldots, s'_{j-1})_p$
	if and only if 
	\markj{$(s_k,z) \in \T\setminus\{s_1,\ldots, s_{k-1}\}$}.
	%$(s_k,z) \in \T(s_1,\ldots, s_{k-1})_p$.
	For $s'_j=s'_1=s_i=s_k$ this is true because none of the elements $s_1,\ldots, s_{i-1}$ are in $N_T(s_i)$ so for each $z$ we have $(s'_1,z)\in \T$ if and only if 
	\markj{$(s_i,z)\in \T\setminus\{s_1,\ldots,s_{i-1}\}$}.
	% $(s_i,z)\in \T(s_1,\ldots,s_{i-1})$.

	For $k$ with $k<i$ we have $s'_{j+1}=s_j$.
	Because $s_i\notin N_T(s_j)$ we have that $(s_j,z) \in \T\setminus \{s'_1,\ldots, s'_j\}=\{s_1,\ldots, s_{j-1},s_i\}$ if and only if $(s_j,z) \in \T\setminus \{s_1,\ldots, s_{j-1}\}$.

	For $k>i$ we have $j=k$ and also $\T\setminus \{s'_1, \ldots s'_{j-1}\}=\T\setminus \{s_1, \ldots s_{j-1}\}$ because $\{s_1, \ldots s_{j-1}\}=\{s'_1, \ldots s'_{j-1}\}$.
	It directly follows that 
	\markj{$(s'_j,z)\in \T\setminus\{s'_1, \ldots s'_{j-1}\}$}
	%$(s'_j,z)\in \T(s'_1, \ldots s'_{j-1})$ 
	if and only if 
	\markj{$(s_j,z)\in \T\setminus\{s_1, \ldots s_{j-1}\}$.}
	%$(s_j,z)\in \T(s_1, \ldots s_{j-1})$.

	Now because we know that for each $k$ we have \markj{$s_k\in H(T\setminus\{s_1,\ldots, s_{k-1}\})$}
    %$s_k\in H(T(s_1,\ldots, s_{k-1}))$
	and $s_k=s'_j$ is in exactly the same cherries in 
	\markj{$T\setminus\{s_1,\ldots, s_{k-1}\}$} 
	%$T(s_1,\ldots, s_{k-1})$
	as in \markj{$T\setminus\{s'_1,\ldots, s'_{j-1}\}$}, 
	% $T(s'_1,\ldots, s'_{j-1})$, 
	we know that 
	\markj{$s'_j\in H(T\setminus\{s'_1,\ldots, s'_{j-1}\})$}, 
	%$s'_j\in H(T(s'_1,\ldots, s'_{j-1}))$, 
	that 
	\markj{$w_{T\setminus\{s'_1,\ldots, s'_{j-1}\}}(s'_j)=w_{T\setminus\{s_1,\ldots, s_{k-1}\}}(s_k)$} 
	%$w_{T(s'_1,\ldots, s'_{j-1})}(s'_j)=w_{T(s_1,\ldots, s_{k-1})}(s_k)$ 
	and that $s'$  satisfies $C$.
	This implies that $s'$ is a CPS with $w_T(s)=w_T(s')$.

\end{proof}

As soon as we know that a leaf in $H(T)$ has to be picked before all its neighbors we can pick it, as stated by the following lemma.
\begin{lem}
	Suppose $x\in H(T)$ and constraint set $C$ is satisfied by cherry picking sequence $s$ of $T$, with $\{(x,n):n\in N_T(x) \}\subseteq C$.
	Then there is a cherry picking sequence $s'$ with $s'_1 = x$ and $w(s')=w(s)$.
	\label{lem:remove_immediately} %\todo[color=green!40]{$N_T(x) \setminus \{x\}$ should be $N_T(x)$?}
\end{lem}
\begin{proof}
	This follows from \cref{lem:twoconditions}, because statement (2) can not be true because for every $j$ with $s_j\in N_T(x)$ we have $(x, s_j)\in C$ and therefore $i < j$ for $s_i=x$.
	So statement (1) has to hold which yields a sequence $s'$ with $w(s)=w(s')$ and $s'_1=x$.
\end{proof}

%\todo[color=green!40]{add sentence introducing next lemma}
The following lemma shows that we can also safely remove all leaves that are in a cherry with the same leaf in every tree.
\begin{lem}
	Let $s$ be a cherry picking sequence for $T$ satisfying constraint set $C$ with $x\notin \pi_1(C)$ and $x\notin \pi_2(C)$.
	If $x\in H(T)$ and $w_T(x)=0$, then there is a cherry picking sequence $s'$ with $s'_1 = x$ and $w(s')=w(s)$ satisfying $C$.
	\label{lem:remove_trivial}
\end{lem}
\begin{proof}
	Because $w_T(x)=0$ we have $N_{T}(x)=\{y\}$.
	Then from \cref{lem:twoconditions} it follows that a sequence $s'$ exists such that either $s''=(x)|s'$ or $s''=(y)|s'$ is a cherry picking sequence for $T$ and $w_T(s'')=w(s)$ and $s''$ satisfies $C$.
	However, because the position of $x$ and $y$ in the trees are equivalent (i.e. swapping $x$ and $y$ does not change $T$) both are true.
\end{proof}

\markj{We are almost ready to describe our algorithm. There is one final piece to introduce first: the measure $P(C)$. 
This is a measure on a set of constraints $C$, which will be used to provide a termination condition for our algorithm. 
We show below that $P(C)$ provides a lower bound on the weight of any cherry picking sequence satisfying $C$, and so if during any recursive call to the algorithm $P(C)$ is greater than the desired weight, we may stop that call.}

%-- moving this bit --
%The algorithm that we present is a recursive algorithm that is called with parameters $k$ and $C$.
%One call to the algorithm results in at most $3$ subcalls.
%In each subcall $2k-|C|$ decreases by at least one.
%It can be shown that for $2k-|C|<0$ no feasible solution exists.
%So we could write the running time $T(2k-|C|)$ in a recurrence equation like $T(2k-|C|)=3\cdot  T(2k-|C|-1) + poly(n,m)$. Solving this equation for \markj{$T(2k-|C|)=x^{2k-|C|}$} yields $x=3$, which gives an $O(9^k\cdot poly(n,m))$ bound on the running time. 
%%\todo[color=green!40]{give short intuitive proof.}
%%\todo{also need to show that we can stop when $2k-|C| < 0$.}
%%\todo[color=yellow]{Do we? We prove a stronger statement in \cref{lem:bound_w_pc}}
%However using the definition of $P(C)$ below we can show that $k-P(C)$ also decreases on each recursive subcall, \markj{and that no feasible solution exists for $k - P(C) < 0$.}
%This allows us to derive a better bound on the running time of $O(5^k\cdot poly(n,m))$.

\label{section:the_algorithm}
\begin{defn}
	Let $\psi = \frac{\log(2)}{\log(5)} \markj{\simeq 0.4307}$.
	Let $P(C) = \psi \cdot |C| + (1-2\psi) |\pi_1(C)|$.
\end{defn}

% -- also moving this --
%Each call to the algorithm branches into either two or three subcalls and reduces the value of $k-P(C)$.
%After picking a leaf, all related constraints are removed from $C$, but also $k$ is reduced by its weight.
%Hence picking a leaf does not increase $k-P(C)$.
%We will prove that as soon as $k-P(C)<0$ no feasible solution exists, so the algorithm then terminates.
%This allows us to derive a bound of $O(5^{k-P(C)} \cdot k\cdot n \cdot m )$ on the running time of the algorithm.
%By calling the algorithm with $C=\emptyset$ this gives a bound on the running time of $O(5^k\cdot k\cdot n \cdot m)$.

\begin{lem}
	If cherry picking sequence $s$ for $T$ satisfies $C$, then $w_T(s)\geq P(C)$.
	\label{lem:bound_w_pc}
\end{lem}
\begin{proof}
	For $x=s_i$ with $i<n$ we prove that for $C_x:=\{(a,b):(a,b)\in C \land a=x \}$ we have $w_{T\setminus \{s_1, \ldots, s_{i-1} \}}(x)\geq P(C_x)$.
	If $|C_x|=0$, then $P(C_x)=0$ and the inequality is trivial.
	If $|C_x|=1$, then there is some $(x,b)\in C$, which implies that $w_{T\setminus \{s_1, \ldots, s_{i-1} \}}(x)>0$, so $w_{T\setminus \{s_1, \ldots, s_{i-1} \}}(x)\geq  |\pi_1(C_x)|=1\geq P(C)$.
	Otherwise if $|C_x|\geq 2$, then %$w_{T\setminus \{s_1, \ldots, s_{i-1} \}}(x)=N_T(x)-1\geq |C_x|-1= \psi \cdot  |C_x|-1 + (1-\psi)|C_x| \geq |C_x|-1 + 2(1-\psi)=|C_x|+(1-2\psi)=P(C_x)$.
	\markj{$w_{T\setminus \{s_1, \ldots, s_{i-1} \}}(x)=N_T(x)-1\geq |C_x|-1= \psi \cdot  |C_x|-1 + (1-\psi)|C_x| \geq \psi \cdot |C_x|-1 + 2(1-\psi)= \psi \cdot |C_x|+(1-2\psi)=P(C_x)$.}
	Now the result follows because $w_T(s)=\sum_{i=1}^{n-1}w_{T\setminus \{s_1,\ldots, s_{i-1}\}}(s_i)\geq \sum_{i=1}^{n-1}P(C_{s_i})=P(C)$.
\end{proof}

\markj{We now present our algorithm, which we split into two parts.
The main algorithm is \texttt{CherryPicking}, a recursive algorithm which takes as input parameters a set of trees $T$, a desired weight $k$ and a set of constraints $C$,
and returns a cherry picking sequence for $T$ of weight at most $k$ satisfying $C$, if one exists.}

%An import part of the algorithm is the \texttt{Pick} procedure. In the
\markj{
The second part is the procedure \texttt{Pick}. In this}
procedure zero-weight cherries and cherries for which all neighbors are contained in the constraint set are greedily removed from the trees.

\begin{algorithm}
	\caption{}
	\label{alg:better_temporal}
	\begin{algorithmic}[1]
		\Procedure{CherryPicking}{$T,k, C$}
		\If {$k-P(C) < 0$} \label{line:if_smaller_pc}
		\State \Return $\emptyset$ \label{line:first_return}
		\EndIf
		\State $T',k', C',p\gets $\Call{Pick}{$T,k,C$} \label{line:callpick}
		\If{$|\mathcal{L}(T')| =1  $}
		\State\Return $\{p\}$ \label{line:emptysequence}
		\ElsIf{$\pi_1(C')\nsubseteq \mathcal{L}(T') $}
		\State \Return $\emptyset$  \label{line:second_return}
		\ElsIf{$k'-P(C')\leq 0$} \label{line:if_smaller_pc_2}
		\State \Return $\emptyset$  \label{line:third_return}
		\EndIf

		\\
		\State $R\gets \emptyset$
		\If{$\exists (x, y) \in G(T',C'): w_T(x)>0 \land x\in \pi_1(C')$ } \label{line:if-statement}
		\State $R\gets R \cup $ \Call{CherryPicking}{$T'$,$k'$,$C'\cup \{ (x,y ) \}$}
		\State $R\gets R \cup $ \Call{CherryPicking}{$T'$,$k'$,$C'\cup \{ (y,x ) \}$}

		\ElsIf{$\exists (x,a) \in G(T',C'): w_{T'}(x)>0 \land x\notin \pi_2(C')$ }  \label{line:else-if-statement}
		\State \text{Choose} $b\neq a$ such that $(x,b)\in G(T',C')$
		\label {line:two_elements}
		\State $R\gets R \cup $ \Call{CherryPicking}{$T'$,$k'$,$C'\cup \{ (a,x ) \}$}
		\State $R\gets R \cup $ \Call{CherryPicking}{$T'$,$k'$,$C'\cup \{ (b,x ) \}$}
		\State $R\gets R \cup $ \Call{CherryPicking}{$T'$,$k'$,$C'\cup \{ (x,a ), (x,b) \}$}
		\EndIf
		\State\Return $\{p|r: r \in R\}$
		\EndProcedure
	\end{algorithmic}
\end{algorithm}
\begin{algorithm}
	\caption{}
	\begin{algorithmic}[1]
		\Procedure{Pick}{$T', k', C'$}
		\State $(T^{(0)},k_1,C_1)\gets (T',k',C')$
		\State $p^{(0)}\gets ()$
		\State $i\gets 1$
		\While{$\exists x_i\in H(T^{(i-1)}):  w_{T^{(i-1)}}(x)=0 \lor \{(x_i,n) : n \in N_{T^{(i-1)}}(x_i) \subseteq C_i $}
		\State $p^{(i)}\gets p^{(i-1)}|(x_i)$
		\State $k_{i}\gets k_{i-1}-w_{T^{(i-1)}}(x_i)$
		\State $T^{(i)}\gets  T^{(i-1)}\setminus \{x_i\}$
		\State $C_{i}\gets \{(a,b)\in C_{i-1}: a\neq x_i\}$
		\State $i\gets i+1$
		\EndWhile
		\State \Return $T^{(i-1)},k_{i-1},C_{i-1},p^{(i-1)}$
		\EndProcedure
	\end{algorithmic}
\end{algorithm}

%\todo[color=green!40]{$T^{(i)}\setminus \{x_i\}$ should be $T^{(i)}$ in algorithm. ALso, split the algorithm in two.}

\subsection{Proof of correctness}
%\todo{make story}
In this section a proof of correctness will be given. First some properties of the auxiliary procedure \texttt{Pick} are proven.
\begin{obs}
Suppose \texttt{Pick}$(T',k', C')$ returns $(T,k,C,p)$.
\begin{enumerate}
    \item There are no $x\in H(T)$ with $w_T(x)=0$.
	\item There are no $x\in H(T)$ with $\{(x_i,n) : n \in N_{T^{(i-1)}}(x_i)\}\subseteq C$.
\end{enumerate}
\end{obs}
\begin{lem}[Correctness of \texttt{Pick}] Suppose \texttt{Pick}$(T',k', C')$ returns $(T,k,C,p)$.
	\begin{enumerate}
		\item If a cherry picking sequence $s$ of weight at most $k$ for $T$ that satisfies $C$ exists then a cherry picking sequence $s'$ of weight at most $k'$ for $T'$ that satisfies $C'$ exists.
		\item 	If $s$ is a cherry picking sequence of weight at most $k$ for $T$ that satisfies $C$ then $p|s$ is a cherry picking sequence for $T'$ of weight at most $k'$ and satisfying $C'$.
	\end{enumerate}
	\label{lem:correctness_pick}
\end{lem}
\begin{proof}
	We will prove the first claim for $(T,k,C,p)=(T^{(i)}, k_i, C_i, p^{(i)})$ for all $i$ defined in \texttt{Pick}.
	We will prove this with induction on $i$.
	For $i=1$ this is obvious because $T^{(1)}=T$, $p^{(1)}=()$, $C_1= C$ and $k_1=k$.

	Now assume the claim is true for $i=i'$.
	Now there are two cases to consider: 

	\begin{itemize}
		\item If we have $\{(x_{i'},n):n\in N_T(x_{i'})  \}\subseteq C_{i'}$ we know from \cref{lem:remove_immediately} that if a cherry picking sequence $s$ satisfying $C_i$ exists then also a cherry picking sequence $(x)|s'$ that satisfies $C'$ exists with $w(p|(x)|s')=w(p|s)$.
		      Note that this implies that $s'$ is a cherry picking sequence for $T^{(i+1)}=T'\setminus \{x\}$, that $C_{i+1}={c\in C': x \notin \{c_1, c_2\}}$ is satisfied by $s_{i+1}$ and that $w(s_{i+1})=w(s_i)-w_{T^{(i)}}(x_i)=k_i-w_{T^{(i)}}(x)$.
		      So this proves the statement for $i=i'+1$.
		\item Otherwise we have $w_{T^{(i')}}(x)=0$ and $x\notin \pi_1
			      (C)$ and $x\notin \pi_2(C)$.
		      Then the statement for $i=i'+1$ follows directly from \cref{lem:remove_trivial}.
	\end{itemize}

	Let $j$ be the maximal value such $x_j$ is defined in a given invocation of \texttt{Pick}.
	
	We will prove the second claim for $(T,k,C,p)=(T^{(i)}, k_{i}, C_{i}, p^{(i)})$ for all $i=0,\ldots, j$ with induction on $i$.
	For $i=0$ this is trivial.
	Now assume the claim is true for $i=i'$ and assume $s$ is a cherry picking sequence for $T^{(i'+1)}$ of weight at most $k_{i'+1}$ that satisfies $C_{i'+1}$.
	Then if $x_{i'}$ is defined, it will be in $H(T^{(i')})$, so $s'=(x_{i'})|s$ is a cherry picking sequence for $T^{(i')}$.
	Because $w_{T^{i'}}(x_{i'})=k_{i'}-k_{i'+1}$, $s'$ will have weight at most $k_{i'}$.
	We can write $C_{i'}=C_x\cup C_{-x}$ where $C_x = \{(a,b):(a,b)\in C_{i'} \land a=x \}$ and $C_{-x}=C_{i'}\setminus C_x$.
	Note that $s$ satisfies $C_{i'+1}=C_{-x}$, so $s'=(x_{i'})|s$ also satisfies $C_{i'+1}$.
	Because for every $(a,b)\in C_x$, also $(a,b)\in T^{i'}$, $s'$ also satisfies $C_{x}$, so $s'$ satisfies $C_{i'}$.
	Now it follows from the induction hypothesis that $p^{i'+1}|s=p^{i'}|s'$ is a cherry picking sequence for $T'$ of weight at most $k'$ and satisfying $C'$.

\end{proof}

%\FloatBarrier
Note that on \cref{line:two_elements} \markj{of \cref{alg:better_temporal}} an element $b\neq a$ with $(x,b) \in G(T', C')$ is chosen. The following lemma states that such an element does indeed exist.
\begin{lem}
	When the algorithm executes \cref{line:two_elements} there exist an element $b\neq a$ with $(x,b) \in G(T', C')$. 
	%\todo{move down}
	\label{lem:exists_uncovered_pair}
\end{lem}
\begin{proof}
	Because $w_{T'}(x)>0$, there is at least a $b\neq a$ such that $b\in N_{T'}(x)\setminus \{x\}$.
	Because $x \notin \pi_2(C')$ we have $(b, x)\notin C'$.
	If $(x,b)\in C'$ then $x\in \pi_1(C')$, but then $x$ satisfies the if-statement on \cref{line:if-statement} and it would not have gotten to this line.
	Therefore $(x,b)\notin C'$ and so $(x,b)\in G(T',C')$.
\end{proof}

The proof of correctness of \cref{alg:better_temporal} will be given in two parts. First, in \label{lem:procedure_returns_sequence} we show that for any feasible problem instance the algorithm will return a sequence. Second, in \label{lem:returned_sequences_satisfy_demands} we show that every sequence that the algorithm returns is a valid cherry picking sequence for the problem instance.
\begin{lem}
	When a cherry picking sequence of weight at most $k$ that satisfies $C$ exists, \texttt{CherryPicking}$(T,k, C)$ from  \Cref{alg:better_temporal} returns a non-empty set.
	\label{lem:procedure_returns_sequence}
\end{lem}
\begin{proof}
	Let $W(k,u)$ be the claim that if a cherry picking sequence $s$ of weight at most $k$ exists that satisfies constraint set $C$ with $n^2-|C|\leq u$, then \leo{calling} \texttt{CherryPicking}$(T,k, C)$ will return a non-empty set.
	We will prove this claim with induction on $k$ and $n^2-|C|$.

	For the base case $k=0$ if a cherry picking sequence of weight $k$ exists we must have that all trees are equal, so $|\mathcal{L}(T)|=1$.
	In this case a sequence is returned on \cref{line:emptysequence}.

	Note that we can never have a constraint set $C$ with $|C|>n^2$ because $C\subseteq \mathcal{L}(T)^2$.
	Therefore $W(k,-1)$ is true for all $k$.

	Now suppose $W(k, n^2-|C|)$ is true for all cases where $0\leq k< k_b$ and all cases where $k=k_b$ and $n^2-|C| \leq u$.
	We consider the case where a cherry picking sequence $s$ of weight at most $k=k_b+1$ exists for $T$ that satisfies $C$ and $n^2-|C|\leq u+1$.
	\Cref{lem:bound_w_pc} implies that $k-P(C)\geq 0$, so the condition of the if-statement on \cref{line:if_smaller_pc} will not be satisfied.

	From \cref{lem:correctness_pick} it follows that a CPS $s'$ of weight at most $k'$ exists for $T'$ that satisfies $C'$.
	From the way the \texttt{Pick} works it follows that either $k'<k$ or $n^2-C'= n^2-C$.
	If $|\Leav(T')=1$ then $\{()\}$ is returned and we have proven $W(k_b+1, u+1)$ to be true for this case.
	Because $s'$ satisfies $C'$, we know that $\pi_1(C) \subseteq \Leav(T')$.
	We know there is an $y\in N_{T'}(s'_1)$ with $(s'_1,y)\notin C'$, because otherwise $s'_1$ would be picked by \texttt{Pick}.
	Also $s'$ satisfies $C'\cup \{(s'_1,y)\}$, which implies that $k\geq P(C'\cup \{(s'_1,y)\})>P(C')$, so the condition of the if-statement on \cref{line:if_smaller_pc_2} will not be satisfied.

	Note that we have $(s'_1, x) \in G(T',C')$, $w_{T'}(s'_1)>0$ and $s'_1\notin \pi_2(C')$.

	This implies that either the body of the if-statement on \cref{line:if-statement} or the body of the else-if-statement on \cref{line:else-if-statement} will be executed.

	Suppose the former is true.
	By \cref{lem:branch_in_two} we know that $s$ satisfies $C'\cup \{(x,y)\}$ or $C'\cup \{(y,x)\}$.
	Because $(x,y)\in G(T',C')$ we know $|C'\cup \{x,y\}|=|C'\cup \{y,x\}|=|C'|+1$ and therefore $n^2-|C'\cup \{x,y\}|=n^2-|C'\cup \{y,x\}| \leq u$.
	So by our induction hypothesis we know that at least one of the two subcalls will return a sequence, so the main call to the function will also return a sequence.

	If instead the body of the else-if-statement on line \cref{line:else-if-statement} is executed we know by \cref{lem:branch_in_three} that at least one of the constraint sets $C'_1=C\cup \{(a,x)\}$,  $C'_2=C\cup \{(b,x)\}$ and $C'_3=C\cup \{(x,a),(x,b)\}$ is satisfied by $s$.
	Note that $|C'_3|\geq |C'_2|=|C'_1|\geq |C'| +1$, so $n^2-|C'_3|\leq n^2-|C'_2|=n^2-|C'_1|\leq u$.
	By the induction hypothesis it now follows that at least one of the three subcalls will return a sequence, so the main call to the function will also return a sequence.
	So for both cases we have proven $W(k_b+1, u+1)$ to be true.
\end{proof}

\begin{lem}
	Every element in the set returned by \texttt{CherryPicking}$(T,k, C)$ from  \Cref{alg:better_temporal} is a cherry picking sequence for $T$ of weight at most $k$ that satisfies $C$.
	\label{lem:returned_sequences_satisfy_demands}
\end{lem}
\begin{proof}
	Consider a certain call to \texttt{CherryPicking}$(T,k, C)$.
	Assume that the lemma holds for all subcalls to \texttt{CherryPicking}.
	We claim that during the execution every element that is in $R$ is a partial cherry picking sequence for $T'$ of weight at most $k'$ that satisfies $C'$.
	This is true because $R$ starts as an empty set, so the claim is still true at that point.
	At each point in the function where sequences are added to $R$, these sequences are elements returned by \texttt{CherryPicking}($T',k', C''$) with $C'\subseteq C''$.
	By our assumption we know that all of these elements are cherry picking sequences for $T'$ of weight at most $k'$ and satisfy $C''$.
	The latter implies that every elements also satisfies $C'$ because $C'\subseteq C''$.
	The procedure now return $\{p|r:r\in R \}$ and from \cref{lem:correctness_pick} it follows that all elements of this set are cherry picking sequences for $T$ of weight at most $k$ and satisfying $C$.
\end{proof}

\subsection{Runtime analysis}

\markj{The key idea behind our runtime analysis is that at each recursive call in \cref{alg:better_temporal2}, the measure $k-P(C)$ is decreased by a certain amount, and this leads to a bound on the number of times \cref{alg:better_temporal} is called.
It is straightforward to get a bound of $O(9^k)$. Indeed, 
it can be shown that for $k<|C|/2$ no feasible solution exists, and so the algorithm could stop whenever $2k - |C| < 0$. 
One call to the algorithm results in at most $3$ subcalls, and
in each subcall $|C|$ increases by at least one.
Then the total number of subcalls to \cref{alg:better_temporal} would be bounded by $O(3^{2k}) = O(9^k)$.
By more careful analysis, and using the lower bound of $P(C)$ on the weight of a sequence satisfying $C$, we are able to improve this bound to $O(5^k)$.
}

We will now state some lemmas that are needed for the runtime analysis of the algorithm. 
\markj{We first show that the measure $k - P(C)$ will never increase at any point in the algorithm. The only time this may happen is during \texttt{Pick}, as the values of $k$ and $C$ are not otherwise changed, except at the point of a recursive call where constraints are added to $C$ (which cannot increase $P(C)$).
Thus we first show that \texttt{Pick} cannot cause $k - P(C)$ to increase.}

\begin{lem}
	\label{lem:pick_decrease_k}
	Let $(s,T',k',C')=\texttt{Pick}(T,k,C)$  from \cref{alg:better_temporal}.
	Then $k'-P(C')\leq k-P(C)$.
\end{lem}
\begin{proof}
	We will prove with induction that for the variables $k_i$ and $C_i$ defined in the function body, we have $k_i-P(C_i) \leq k-P(C)$ for all $i$, from which the result follows.
	Note that for $i=0$ this is trivial.
	Now suppose the inequality holds for $i$.
	Then we also have
	\begin{align*}
		k_{i+1} - P(C_{i+1}) & = (k_{i} - w_{T^{(i)}}(x_i)) - (P(C_i)- (w_{T^{(i)}}(x_i)+1) \cdot \psi - (1-2\psi)) \\ &= k_i - P(C_i) - (w_{T^{(i)}}(x_i) - 1) (1-\psi) \\ &\leq k_i - P(C_i) \\ &\leq k - P(C)
	\end{align*}
\end{proof}

\FloatBarrier
\markj{The next lemma will be used later to show that a recursive call to 	\texttt{CherryPicking} always increases $k-P(C)$ b a certain amount.}

\begin{lem}
	For $a$ and $b$ on \cref{line:two_elements} \markj{of \cref{alg:better_temporal}} it holds that $a\notin \pi_1(C')$ \markj{and $b\notin \pi_1(C')$.}
	\label{lem:algline_a_not_before}
\end{lem}
\begin{proof}
	Suppose $a\in \pi_1(C')$.
	Then $(a,z)\in C'$ for some $z\in N_{T'}(x)$.
	If $w_{T'}(a)>0$ then $a$ satisfies the conditions in the if-statement on \cref{line:if-statement}, so \cref{line:two_elements} would not be executed.
	If $w_{T'}(a)=0$ then we must have  $|N_{T'}(a)\setminus \{a\}| = 1$, so $N_{T'}(a)\setminus \{a\} = \{x\}$, which implies that $z=x$.
	But $(a,x)\notin C'$ because $(x,a)\in G(T',C')$, which contradicts that $(a,z)\in C'$.
	So $a\notin \pi_1(C')$.
	Because of symmetry, the same argument holds for $b$.
\end{proof}

\FloatBarrier
\markj{We now give the main runtime proof.}
\begin{lem}
	\texttt{CherryPicking} from \Cref{alg:better_temporal} has a time complexity of $O(5^k \cdot knm)$.
	\label{lem:running_time_better_temporal}
\end{lem} 
\begin{proof}
    \sander{Let $n$ be the number of leaves and $m$ the number of trees.
	The non-recursive part of \texttt{CherryPicking}($T$,$k$,$C$) can be implemented to run in $O(n\cdot m)$ time by constructing $H(T^{(i)})$ from $H(T^{(i-1)})$ in each step.}
	%\todo{I am not convinced about this running time; naively I think the while condition in \texttt{Pick} takes $O(n\cdot m)$ time to check ($O(m)$ to check the condition for $x$, for $O(n)$ values of $x$), and we may need to check it $O(n)$ times, so \Pick takes $O(n^2m)$ time}
	Let $f(n,m)$ be an upper bound for its computation time with $f(n,m)=O(n\cdot m)$.
	Let the runtime of \texttt{CherryPicking}($T$,$k$,$C$) be $t(n, k, C)$.
	We will prove this with induction on $k-P(C)$ that
	\begin{align*}
		%t(n, k, C) \leq 5^{k-P(C)+1} f(n,m)
		\markj{t(n, k, C) \leq 5^{k-P(C)+1} (k-P(C)+1)f(n,m)} \text{.
		}
	\end{align*}
	%\todo[color=green!40]{Previous proof seemed to use $k$ instead of $(k-P(C)+1)$; however this is not guaranteed to decrease at each recursion step}
	For $-1\leq k-P(C)\leq 0$ the claim follows from the fact that the function will return on either \cref{line:first_return} or \cref{line:third_return} and therefore will not do any recursive calls.

	Now assume the claim holds for $-1\leq k-P(C)\leq w$.
	Now consider an instance with $k-P(C)\leq  w+\psi $.
	Note that  $k'-P(C')\leq k-P(C)$  (\cref{lem:pick_decrease_k}).
	If the function \texttt{CherryPicking} does any recursive calls then it either executes the body of the if-clause on \cref{line:if-statement}, or the body of the else-if clause on \cref{line:else-if-statement}.

	If the former is true then the function does $2$ recursive calls.
	Each recursive call \leo{to the function} \texttt{CherryPicking}($T'$, $k'$, $C''$) is done with a constraint set $C''$ for which $|C''|=|C'|+1$.
	Therefore for both subproblems $P(C'')\geq P(C') + \psi$ and also $k'-P(C'')\leq k'-P(C') - \psi \leq  k-P(C) - \psi \leq  w$.
	By our induction hypothesis the running time of each of the subcalls is now bounded by %${5^{k'-P(C'')+1}(k'-P(C'')+1)f(n)}$.
	 \markj{${5^{k'-P(C'')+1}(k'-P(C'')+1)f(n,m)}$.}
	So therefore the total running time of this call is bounded by
	
	\markj{
	\begin{align*}
		 & 2\cdot 5^{k'-P(C'') + 1}(k'-P(C'') + 1)f(n,m) + f(n,m)\\ 
		 &\leq 2\cdot 5^{k-P(C)-\psi +1}(k-P(C)-\psi +1)f(n,m) +f(n,m) \\&=5^\psi 5^{k-P(C)-\psi +1}(k-P(C)-\psi +1)f(n,m)+ f(n,m) \\& = 5^{k-P(C) +1}(k-P(C)-\psi +1) f(n,m)+ f(n,m) \\
		 &\sander{\leq} 5^{k-P(C) +1}(k-P(C)+1)f(n,m) - 5\psi f(n,m) + f(n,m)\\
		 &\leq 5^{k-P(C) +1}(k-P(C)+1) f(n,m)\text{.
		}
		% & 2\cdot 5^{k'-P(C'') + 1}(k-1)f(n,m) + f(n,m) \leq 2\cdot 5^{k-P(C)-\psi +1}(k-1)f(n,m) +f(n,m) \\&=5^\psi 5^{k-P(C)-\psi +1}(k-1)f(n,m)+ f(n,m) \\& = 5^{k-P(C) +1}(k-1) f(n,m)+ f(n,m) \leq 5^{k-P(C) +1}k f(n,m)\text{.
		%}
	\end{align*}}
	So in this case we have proven the claim for $-1\leq k-P(C)\leq  w+\psi $.

	If instead the body of the else-if statement on \cref{line:else-if-statement} is executed then 3 recursive subcalls are made.
	Consider the first subcall $\texttt{CherryPicking}(T',k',C'')$.
	We  have $C''=C'\cup \{(a,x)\}$.
	Because \markj{$(x,a)\in G(T',C')$ we have $(a,x)\notin C'$.}
	Therefore $|C''|=|C'|+1$.
	By \cref{lem:algline_a_not_before} we know that $a\notin \pi_1(C')$, but we have $a\in \pi_1(C')$, so $|\pi_1(C'')|=|\pi_1(C')|+1$.
	Therefore $P(C'')=P(C')+1-\psi$, so $k'-P(C'')=k'-P(C')-1+\psi \markj{\leq k - P(C) -1 + \psi }< k-P(C)-\psi \leq w$.
	By our induction hypothesis we now know that the running time of this subcall is bounded by
	\markj{
	\begin{align*}
		5^{k'-P(C'')+1} (k'-P(C'')+1)f(n,m) \leq 5^{k-P(C)+\psi} (k-P(C)+\psi)f(n,m) \text{.
		}
	\end{align*}}
%	\begin{align*}
%		5^{k'-P(C')+\psi} (k-1)f(n,m) \leq 5^{k-P(C)+\psi} (k-1)f(n,m) \text{.
%		}
%	\end{align*}
	Note that by symmetry the same holds for the second subcall.

	For the third subcall 	$\texttt{CherryPicking}(T',k',C'')$ ,
	\markj{because $(x,a),(x,b) \in G(T',C')$ we have $|C''| = |C'|+2$, and because $x \notin \pi_1(C')$ we have $|\pi_1(C'')| = |\pi_1(C')| + 1$. So}
we know that $P(C'')=\markj{P(C') + 2\psi + (1-2\psi) =} P(C')+1$
\markj{and $k'-P(C'') +1 \leq k-P(C)$.}
	\markj{Therefore the} running time is bounded by
	\markj{\begin{align*}
		5^{k-P(C)}(k-P(C))f(n,m)\text{.
		}
	\end{align*}}
	%\begin{align*}
	%	5^{k-P(C)}(k-1)f(n,m)\text{.
	%	}
	%\end{align*}
	So the total running time of this call is bounded by
	\markj{
	\begin{align*}
		 & \hphantom{= =} 2\cdot 5^{k-P(C)+\psi} (k-P(C)+\psi)f(n,m)+ 5^{k-P(C)}(k-P(C))f(n,m)+f(n,m) \\
		 & = 2\cdot 5^\psi \cdot 5^{k-P(C)} (k-P(C)+\psi)f(n,m)+ 5^{k-P(C)}(k-P(C))f(n,m)+f(n,m) \\
		 & = 4 \cdot 5^{k-P(C)} (k-P(C)+\psi)f(n,m)+ 5^{k-P(C)}(k-P(C))f(n,m)+f(n,m) \\
		 & = 5 \cdot 5^{k-P(C)} (k-P(C))f(n,m)+ 4\cdot\psi\cdot5^{k-P(C)}f(n,m)+f(n,m) \\
		 & \leq 5 \cdot 5^{k-P(C)} (k-P(C))f(n,m)+ 5\cdot5^{k-P(C)}f(n,m) \\
		 & = 5\cdot 5^{k-P(C)} (k-P(C)+1)f(n,m)\\
		 & = 5^{k-P(C)+1} (k-P(C)+1)f(n,m)
	\end{align*}}
	%\begin{align*}
	%	 & \hphantom{= =} 2\cdot 5^{k-P(C)+\psi} (k-1)f(n)+ 5^{k-P(C)}f(n,m)+f(n,m) \\& = (2\cdot 5^{\psi} + 1)5^{k-P(C)}(k-1)f(n,m)+f(n,m) \\&= 5\cdot 5^{k-P(C)}(k-1)f(n,m)+f(n,m) = 5^{k-P(C)+1}(k-1)f(n,m)+f(n,m)\\&\leq 5^{k-P(C)+1}k\cdot f(n,m)
	%\end{align*}
	So also for this case we have proven the claim for $k-P(C)\leq  w+\psi $.
\end{proof}

\begin{thm}
	\texttt{CherryPicking}$(T,k, C)$ from  \Cref{alg:better_temporal} returns a cherry picking sequence of weight at most $k$ that satisfies $C$ if and only if such a sequence exists.
	The algorithm terminates in $O(5^k\cdot poly(n,m))$ time.
\end{thm}
\begin{proof}
	This follows directly from \cref{lem:returned_sequences_satisfy_demands}, \cref{lem:procedure_returns_sequence} and \cref{lem:running_time_better_temporal}.
\end{proof}

% ---------- NON TEMPORAL ------------

\section{Constructing non-temporal tree-child networks from binary trees}
\label{sec:non_temporal}
For every set of trees there exists a tree-child network that displays the trees.
However there are sets of trees for which no temporal network displaying the trees exist, so we can not always find such a network.
As shown in \cref{fig:temporal_tree_child_difference}, approximately 5 percent of the instances used in \cite{van_iersel_practical_2019} do not admit a temporal solution.

\begin{figure}
	\includegraphics[scale=.8]{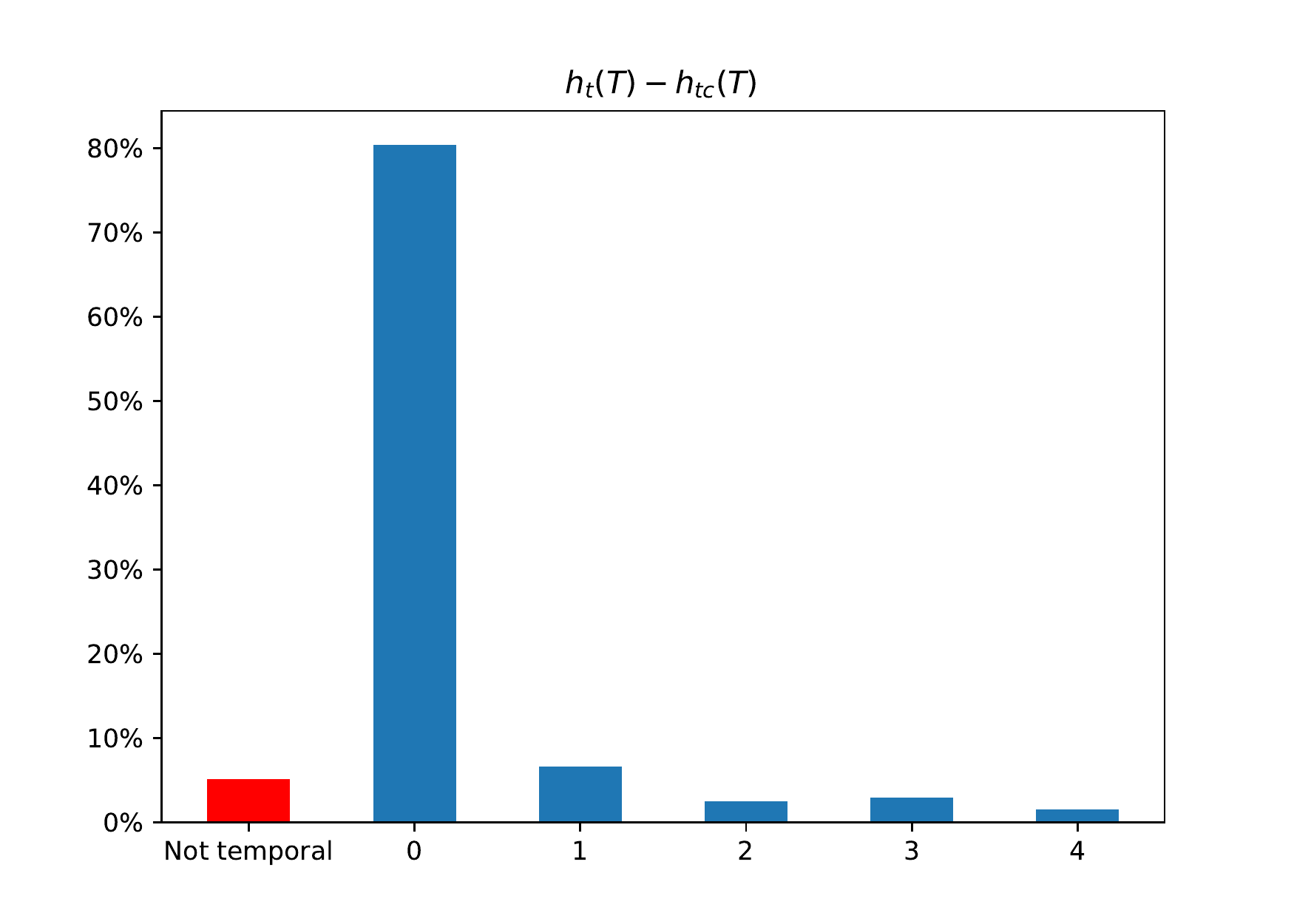}
	\caption{The difference between the tree-child reticulation number and the temporal reticulation number on the dataset generated in \cite{van_iersel_practical_2019}.
		If no temporal network exists, the instance is shown under `Not temporal'.
		Instances for which it could not be decided if they were temporal within 10 minutes ($2.6\%$ of the instances), are excluded.
		%30 of 1155
		\label{fig:temporal_tree_child_difference}
	}

\end{figure}

In this section we introduce theory that makes it possible to quantify how close a network is to being temporal.
We can then pose the problem of finding the `most' temporal network that displays a set of trees.

\begin{defn}
	For a tree-child network with vertices $V$ we call a function $t: V\to \mathbb{R}^+$ a semi-temporal labeling if:
	\begin{enumerate}
		\item For every tree arc $(u,v)$ we have $t(u)<t(v)$.
		\item For every hybridization vertex $v$ we have $t(v)=\min \{t(u): (u,v)\in E\}$.
	\end{enumerate}
\end{defn}
Note that network has a semi-temporal labeling.

\begin{defn}
	For a tree-child network $\mathcal{N}$ with a semi-temporal labeling $t$, define $d(\mathcal{N}, t)$ to be number of hybridization arcs $(u,v)$ with $t(u)\neq t(v)$.
	We call these arcs non-temporal arcs.
\end{defn}
\begin{defn}
	For a tree-child network $\mathcal{N}$ define
	\begin{align*}
		d(\mathcal{N})=\min  \{d(\mathcal{N}, t): t \text{ is a semi-temporal labeling of }\mathcal{N} \}
	\end{align*}
	Call this number the \emph{temporal distance} of $\mathcal{N}$.
	Note that this number is finite for every network, because there always exist semi-temporal labelings.
\end{defn}

The temporal distance is a way to quantify how close a network is to being temporal.
The networks with temporal distance zero are the temporal networks.
We can now state a more general version of the decision problem.

\problem{Semi-temporal hybridization}{A set of~$m$ trees $T$ with~$n$ leaves and integers $k,p$.}{Does there exist a tree-child network $\mathcal{N}$ with $r(\mathcal{N})\leq k$ and $d(\mathcal{N})\leq p$?}

There are other, possibly more biologically meaningful ways to define such a temporal distance.
The reason for defining the temporal distance in this particular way is that an algorithm for solving the corresponding decision problem exists.
For further research it could be interesting to explore if other definitions of temporal distance are more useful and whether the corresponding decision problems could be solved using similar techniques.

Van Iersel et al.
\ presented an algorithm to solve the following decision problem in $O((8k)^k\cdot \text{poly}(m,n))$ time.
\problem{Tree-child hybridization}{A set of~$m$ trees $T$ with~$n$ leaves and integer $k$.}{Does there exist a tree-child network $\mathcal{N}$ with $r(\mathcal{N})\leq k$?}
Notice that for $p=k$ \textsc{Semi-temporal hybridization} is equivalent to \textsc{Tree-child hybridization} and for $p=0$ it is equivalent to \textsc{Temporal hybridization}.
The algorithm for \textsc{Tree-child hybridization} uses a characterization by Linz and Semple \cite{linz_attaching_2019} using \emph{tree-child sequences}, that we will describe in  the next section.
We describe a new algorithm that can be used to decide \textsc{Semi-temporal hybridization}.
This algorithm is a combination of the algorithms for \textsc{Tree-child hybridization} and  \textsc{Temporal hybridization}.
\subsection{Tree-child sequences}
First we will define the \emph{generalized cherry picking sequence} (generalized CPS), which is called a cherry picking sequence in \cite{van_iersel_practical_2019}.
We call it generalized cherry picking sequence because it is a generalization of the cherry picking sequence we defined in \cref{def:cps}.
\begin{defn}
	A \emph{partial generalized CPS} on $X$ is a sequence
	\begin{align*}
		s=((x_1,y_1),\ldots, (x_r,y_r),(x_{r+1},-),\ldots ,(x_{t},-))
	\end{align*}
	with $\{x_1, x_2, \ldots, x_s,y_1,\ldots, y_r \}\subseteq X$.
	A generalized CPS is \emph{full} if $t>r$ and $\{x_1,\ldots, x_t\}=X$.
\end{defn}

For a tree $\T$ on $X'\subseteq X$ the sequence $s$ defines a sequence of trees $(\T^{(0)}, \ldots, \T^{(r)})$ as follows:
\begin{itemize}
	\item $\T^{(0)}=\T$.
	\item If $(x_j,y_j)\in \T^{(j-1)}$, then $\T^{(j)}=\T^{(j-1)} \setminus \{x_j\}$.
	      Otherwise $\T^{(j)}=\T^{(j-1)}$.
\end{itemize}
We will refer to $\T^{(r)}$ as $\T(s)$, the tree obtained by applying sequence $s$ to $\T$.
\iffalse Call a sequence \emph{non-redundant} for $T$ if $T^{(i-1)}=T^{(i)}$ for all $1\leq i\leq r$.
\fi

A full generalized CPS on $X$ is a \emph{generalized CPS} for a set $T$ of trees if for each $\T\in T$ the tree $\T(s)$ contains just one leaf and that leaf is in $\{x_{r+1},\ldots ,x_{t} \}$.
The \emph{weight} of a sequence $s$ for a set of trees on $X$ is defined as $w_T(s)=|s|-|X|$.

A generalized CPS is a \emph{tree-child} sequence if $|s|\leq r+1$ and $y_j\neq x_i$ for all $1\leq i<j\leq |s|$.
If for such a \emph{tree-child} sequence $|s|=r$, then $s$ is also called a \emph{tree-child} sequence prefix.

It has been proven that a tree-child network displaying a set of trees $T$ with $r(\mathcal{N})=k$ exists if and only if a tree-child sequence $s$ with $w(s)=k$ exists.
The network can be efficiently computed from the corresponding sequence.
The algorithm presented by Van Iersel et al.
works by searching for such a sequence.

We will show that it is possible to combine their algorithm with the algorithm presented in \cref{sec:non_temporal}.
This yields an algorithm that decides \textsc{Semi-temporal hybridization} in $O(5^{k}(8k)^p\cdot k\cdot n\cdot m)$ time.

\begin{defn}
	Let $s=((x_1,y_1),\ldots,(x_t,-))$ be a full generalized CPS.
	An element $(x_i,y_i)$ is a \emph{non-temporal} element when there are $j,k\in [t]$ with $i<j<k\leq t$ and $x_j\neq x_i$ and $x_k=x_i$.
\end{defn}
\begin{defn}
	For a sequence $s$ we define $d(s)$ to be the number of non-temporal elements in $s$.
\end{defn}

% Adding commands to allow us to state the same lemma with the same numbering in two different places
\newcommand{\lemSemiTemporalTequenceToNetworkText}{Let $s$ be a full tree-child sequence $s$ for $T$.
	Then there exists a network $\N$ with semi-temporal labeling $t$ such that $r(\N)\leq w_T(s)$ and $d(\N,t)\leq d(s)$.}
\newtheorem*{lemSemiTemporalTequenceToNetwork}{Lemma~\ref{lem:SemiTemporalTequenceToNetwork}}
%First appearance
\begin{lem}
	\label{lem:SemiTemporalTequenceToNetwork}
   \lemSemiTemporalTequenceToNetworkText
\end{lem}
\markj{The full proof of \cref{lem:SemiTemporalTequenceToNetwork} is given in the appendix. We construct a tree-child network $\N$ from $s$ in a similar way to \cite[Proof of Theorem 2.2]{linz_attaching_2019}, working backwards through the sequence. At each stage when a pair $(x,y)$ is processed, we adjust the network to ensure there is an arc from the parent of $y$ to the parent of $x$. 
Our contribution is to also maintain a semi-temporal labeling $t$ on $\N$.
This can done in such a way that for each pair $(x,y)$, at most one new non-temporal arc is created, and only if $(x,y)$ is a non-temporal element of $s$. This ensures that $d(\N,t)\leq d(s)$.}

% Adding commands to allow us to state the same lemma with the same numbering in two different places
\newcommand{\lemSemiTemporalNetworkToSequenceText}{
	For a tree-child network $\N$ there exists a full tree-child sequence $s$ with $d(s)\leq d(\N)$ and $w_T(s)\leq r(\N)$. }
\newtheorem*{lemSemiTemporalNetworkToSequence}{Lemma~\ref{lem:SemiTemporalNetworkToSequence}}
%First appearance
\begin{lem}
	\label{lem:SemiTemporalNetworkToSequence}
   \lemSemiTemporalNetworkToSequenceText
\end{lem}

\markj{
The full proof of \cref{lem:SemiTemporalNetworkToSequence} is given in the appendix. We construct the sequence in a similar way to  \cite[Lemma 3.4]{linz_attaching_2019}. The key idea is that at any point the network will contain some pair of leaves $x,y$ that either form a \emph{cherry} (where $x$ and $y$ share a parent) or a \emph{reticulated cherry} (where the parent of $x$ is a reticulation, with an incoming edge from the parent of $y$).
We process such a pair by appending $(x,y)$ to $s$, deleting an edge from $\N$, and simplifying the resulting network.
By being careful about the order in which we process reticulated cherries, we can ensure that we only add a non-temporal element to $s$ when we delete a non-temporal arc from $\N$.   This ensures that $d(s) \leq d(\N,t)$.}

\begin{obs}
	\label{obs:tc:either_ab_or_ba}
	A tree-child sequence $s$ can not contain both $(a,b)$ and $(b,a)$.
\end{obs}

\begin{obs}
	\label{lem:tree_child_subsequence}
	If a tree-child sequence $s$ has a subsequence $s'$ that is a generalized cherry picking sequence for $T$, then $s$ is also a generalized cherry picking sequence for $T$.
\end{obs}
\begin{lem}
	\label{lem:tc:minus_still_seq}
	If $s=((x_1,y_1),\ldots,(x_{r+1},-))$ is a generalized CPS for $\T$ and there is a $z$ such that $y_i\neq z$ for all $i$.
	Then $(\T\setminus \{z\})(s)=\T(s)$ and therefore $s$ is also a generalized CPS for $\T\setminus \{z\}$.
\end{lem}
\begin{proof}
	Suppose this is not true.
	Because $\T(S)$ consists of a tree with only one leaf $x_{r+1}$, this implies that $\Leav((\T\setminus \{z\})(s))\not \subseteq \Leav(\T(s))$.
	Let $i$ be the smallest $i$ for which \leo{we have that} $\Leav((\T\setminus \{z\})((x_1,y_1),\ldots,(x_i,y_i))) \not\subseteq   \Leav(\T((x_1,y_1),\ldots,(x_i,y_i))\setminus \{z\})$.

	\begin{sloppypar}
		This implies that $x_i \in \Leav((\T\setminus \{z\})((x_1,y_1),\ldots,(x_i,y_i)))$ but $x_i\notin \Leav(\T((x_1,y_1),\ldots,(x_i,y_i))\setminus \{z\})$, so $(x_i,y_i)\notin (\T\setminus \{z\})((x_1,y_1),\ldots,(x_{i-1},y_{i-1}))$, but  $(x_i,y_i)\in \T((x_1,y_1),\ldots,(x_{i-1},y_{i-1}))\setminus\{z\}$.
		Let $p$ be the lowest vertex that is an ancestor of both $x_i$ and $y_i$ in the tree $(\T\setminus~\{z\})((x_1,y_1),\ldots,(x_{i-1},y_{i-1}))$.
		Because $x_i$ and $y_i$ do not form a cherry in this tree, there is another leaf $q$ that is reachable from $p$.
		Because $q\in \Leav(\T((x_1,y_1),\ldots,(x_{i-1},y_{i-1}))\setminus \{z\})$, $q$ is also reachable from the lowest common ancestor $p'$ in $\T((x_1,y_1),\ldots,(x_{i-1},y_{i-1}))\setminus \{z\}$, contradicting the fact that $(x_i,y_i)$ is a cherry in this tree.
	\end{sloppypar}
\end{proof}

\subsection{Constraint sets}
The new algorithm also uses constraint sets.
However, because the algorithm searches for a generalized cherry picking sequence, we need to define what it means for such a sequence to satisfy a constraint set.
\begin{defn}
	A generalized cherry picking sequence $s=((x_1,y_1),\ldots, (x_k,y_k))$ satisfies constraint set $C$ if for every $(a,b)\in C$ there is an $i$ with  $(x_i,y_i)=(a,b)$ and there is some $j\neq i$ with $x_j=a$.
\end{defn}

In \cref{def:ht} the function $H(T)$ was defined for sets of binary trees with the same leaves.
After applying a tree-child sequence not all trees will necessarily have the same leaves.
Because of this, we generalize the definition of $H(T)$ to sets of binary trees.

\begin{defn}
	For a set of binary trees $T$ define $H(T) =\{x\in\Leav(T): \forall \T \in T \text{ if } x\in \T \text{ then }x\text{ is in a cherry in }\T \}$.
\end{defn}

\begin{lem}
	If $s=((x_1,y_1),\ldots,(x_{r+1},-))$ is a tree-child sequence for $T$ and $(a,b)\in T$, then there is an $i$ such that $(x_i,y_i)=(a,b)$ or $(x_i,y_i)=(b,a)$.
	\label{lem:tc:leastonepicked}
\end{lem}
\begin{proof}
	Let $T\in T$ be a tree in $T$ containing cherry $(a,b)$.
	Because $s$ fully reduces $T$, $\T(s)$ consists of only the leaf $x_{r+1}$.
	So $a$ or $b$ has to be removed from $\T$ by applying $s$.
	Without loss of generality we can assume $a$ is removed first.
	This can only happen if there is an $i$ with $(x_i, y_i)=(a,b)$.
\end{proof}

Now we prove that if there are two cherries $(a,z)$ and $(b,z)$ in $T$, then we can branch on three possible additions to the constraint set, just like we did for cherry picking sequences.
\begin{lem}
	Let $s$ be a tree-child sequence for $T$ and $a,b\in N_T(z)$ with $a\neq b$.
	Then $s$ satisfies one of the following constraint sets: \\$\{(a,z)\}, \{(b,z)\}, \{(z,a),(z,b)\}$.
	\label{lem:tc:branch_in_three}
\end{lem}
\begin{proof}
	From \cref{lem:tc:leastonepicked} it follows that either $(a,z)$ or $(z,a)$ is in $s$ and that either $(b,z)$ or $(z,b)$ is in $s$.
	Now let $s_i=(x_i,y_i)$ be the element of these that appears first in $s$.
	Now we have three cases:
	\begin{enumerate}
		\item If $x_i=a$, then $s_i=(a,z)$.
		      Let $\T \in T$ be the tree in which $(b,z)$ is a cherry.
		      Now $(b,z)\in \T(s_1,\ldots, s_i)$.
		      Because $(s_{i+1},\ldots, s_{r+1})$ is a tree-child sequence for $\T(s_1,\ldots, s_i)$, this implies that there is some $j>i$ with $x_j=a$.
		      Consequently $\{(a,z)\}$ is satisfied by $s$.
		\item If $x_i=b$, then the same argument as in (1) can be applied to show that $\{(b,z)\}$ is satisfied by $s$.
		\item If $x_i=z$, then we either have $y_i=a$ or $y_i=b$.
		      Without loss of generality we can assume $y_i=a$.
		      We still have $(b,z)\in T(s_1,\ldots, s_i)$, which implies that there is some $j>i$ with $(x_j,y_j)=(b,z)$ or $(x_j,y_j)=(z,b)$.
		      Because $j>i$ and $s$ is tree-child, we know that $y_j\neq z$.
		      So $(x_j,y_j)=(z,b)$, and consequently $\{(z,a),(z,b)\}$ is satisfied by $s$.
	\end{enumerate}
\end{proof}

We also prove that if $a\in \pi_1(C)$ and $(a,b)\in T$, then we only need to do two recursive calls.
\begin{lem}
	\label{lem:tc:branch_in_two}
	Let $s$ be a tree-child sequence for $T$ that satisfies constraint set $C$ and $a,b\in N_T(z)$ with $(z,b)\in C$.
	Then $s$ satisfies one of the following constraint sets: \\$\{(a,z)\}, \{(z,a)\}$.
\end{lem}
\begin{proof}
	From \cref{lem:tc:branch_in_three} it follows that $s$ satisfies one of the constraint sets $\{(a,z)\}$, $\{(b,z)\}$ and $\{(z,a),(z,b)\}$.
	However, because $s$ satisfies $C$ and $(z,b)\in C$, from \cref{obs:tc:either_ab_or_ba} it follows that $(b,z)$ does not appear in $s$.
	Therefore $s$ has to satisfy either $\{(a,z)\}$ or $\{(z,a),(z,b)\}$.
	If $s$ satisfies $\{(z,a),(z,b)\}$, then it also satisfies $\{(z,a)\}$.
\end{proof}

\begin{lem}
	If a tree-child sequence $s=((x_1,y_1),\ldots,(x_r,x_y),(x_{r+1},-))$ for $T$ satisfies constraint set $C$, then $w_T(s)\geq P(C)$.
	\label{lem:bound_w_pc_nontemp}
\end{lem}
\begin{proof}
	For $z\in \Leav(T)\setminus \{x_{r+1}\}$, let $C_z:=\{(a,b):(a,b)\in C \land a=z \}$ and let $S_z:=\{(x_i,y_i):i\leq r \land x_i=z \}$.
	We show that we have $|S_z|-1\geq P(C_x)$.
	If $|C_z|=0$, then $P(C_z)=0$ and the inequality is trivial.
	If $|C_z|=1$, then from the definition of constraint sets it follows that $|S_z|\geq 2$, so $|S_z|-1\geq 1\geq P(C_z)$.
	Otherwise if $|C_z|\geq 2$, then because $C_z\subseteq S_z$, $|S_z|-1\geq |C_z|-1= \psi \cdot  |C_z|-1 + (1-\psi)|C_z| \geq |C_z|-1 + 2(1-\psi)=|C_z|+(1-2\psi)=P(C_z)$.
	Now the result follows because $w_T(s)=|s|-|\Leav(T)|=\sum_{z\in \Leav(T)\setminus \{x_{r+1}\}}(|S_z|-1)\geq \sum_{z\in \Leav(T)\setminus \{x_{r+1}\}}P(C_z)=P(C)$.
\end{proof}
Next we prove that if a leaf $z$ is in $H(T)$ and appears in $s$ with all of its neighbors, then we can move all elements containing $z$ to the start of the sequence.
\begin{lem}
	\label{lem:tc:remove_immediately}
	If $s=((x_1,y_1),\ldots,(x_{r+1},-))$ is a tree-child sequence for $T$, $z\in H(T)$ and $I$ is a set of indices such that $\{y_i: i\in I \}=N_T(z)$ and $x_i=z$ for all $i\in I$.
	Then the sequence $s'$ obtained by first adding the elements from $s$ with an index in $I$ and then adding elements $(x,y)$ of $s$ for which $x\neq z$ is a tree-child sequence for $T$.
	We have $d(s')\leq d(s)$.
\end{lem}
\begin{proof}
	We can write $s'=((x'_1,y'_1),\ldots,(x_{r+1},-))=s^a|s^b$ where $s^a$ consists of the elements $\{s_i:i\in I\}$ and $s^b$ is $s$ with the elements at indices in $I$ removed. First we prove that $s'$ is a tree-child sequence.
	Suppose that $s'$ is not a tree-child sequence.
	Then there are $i,j$ with $i<j$ such that $x'_i = y'_j$.
	Note that we can not have that $y'_j=z$, because of how we constructed $s'$.
	This implies that both indices $i$ and $j$ are in $s^b$, implying that $s^b$ is not tree-child.
	But because $s^b$ is a subsequence of $s$ this implies that $s$ is not tree-child, which contradicts the conditions from the lemma.
	So $s'$ is tree-child.

	We now prove that $s'$ fully reduces $T$.
	Because $T(s^a)=T\setminus \{z\}$ from \cref{lem:tc:minus_still_seq} it follows that $s^a|s$ is a generalized CPS for $T$.
	Because $z\notin \Leav(T(s^a))$, $T(s^a|s)=T(s^a|s^b)$.
	So $s'$ is a generalized CPS for $T$.

	Finally since for every non-temporal element in $s'$ the corresponding element in $s$ is also non-temporal.
	We conclude that $d(s')\leq d(s)$.

\end{proof}

\subsection{Trivial cherries}
We will call a pair $(a,b)$ a \emph{trivial cherry} if there is a $\T\in T$ with $a\in \Leav(\T)$ and for every tree $\T\in T$ that contains $a$, we have $(a,b)\in \T$.
They are called trivial cherries because they can be picked without limiting the possibilities for the rest of the sequence, as stated in the following lemma.

\begin{lem}
	\label{lem:tc:trivial_cherries}
	If $s=((x_1,y_1),\ldots,(x_{r+1},-))$ is a tree-child sequence for $T$ of minimum length and $(a,b)$ is a trivial cherry in $T$, then there is an $i$ such that $(x_i,y_i)=(a,b)$ or $(x_i,y_i)=(b,a)$.
	Also, there exists a tree-child sequence $s'$ for $T$ with $|s|=|s'|$, $d(s')=d(s)$ and $s'_1=(a,b)$.
\end{lem}
\begin{proof}
	This follows from \cref{lem:tc:remove_immediately}.
\end{proof}

\begin{algorithm}
	\caption{}
	\label{alg:better_temporal2}
	\begin{algorithmic}[1]
		\Procedure{SemiTemporalCherryPicking}{$T,k, k^\star, p, C$}
		\If {$k-P(C) < 0$} \label{line2:first_if}
		\State \Return $\emptyset$ \label{line2:first_return}
		\EndIf
		\State $T',k', C',f\gets $\Call{Pick}{$T,k,C$} \label{line2:callpick}
		\If{$|\mathcal{L}(T')| =1  $}
		\State\Return $\{f\}$ \label{line2:emptysequence}
		\ElsIf{$k'-P(C')\leq 0 \lor \pi_1(C')\nsubseteq \mathcal{L}(T') $}\label{line2:third_if}
		\State \Return $\emptyset$  \label{line2:second_return}
		\EndIf
		\\
		\State $R\gets \emptyset$
		\If{$\exists (x, y) \in T: w_T(x)>0 \land x\in \pi_1(C')$ } \label{line2:if-statement}
		\State $R\gets R \cup $ \Call{SemiTemporalCherryPicking}{$T'$, $k'$, $k^\star$, $p$, $C'\cup \{ (x,y ) \}$}
		\State $R\gets R \cup $ \Call{SemiTemporalCherryPicking}{$T'$, $k'$, $k^\star$, $p$, $C'\cup \{ (y,x ) \}$}

		\ElsIf{$\exists (x,a) \in G(T',C'): w_{T'}(x)>0 \land x\notin \pi_2(C')$ }  \label{line2:else-if-statement}
		\State \text{Choose} $b\neq a$ such that $(x,b)\in G(T',C')$
		\label {line2:two_elements}
		\State $R\gets R \cup $ \Call{SemiTemporalCherryPicking}{$T'$, $k'$, $k^\star$, $p$, $C'\cup \{ (a,x ) \}$}
		\State $R\gets R \cup $ \Call{SemiTemporalCherryPicking}{$T'$, $k'$, $k^\star$, $p$, $C'\cup \{ (b,x ) \}$}
		\State $R\gets R \cup $ \Call{SemiTemporalCherryPicking}{$T'$, $k'$, $k^\star$, $p$, $C'\cup \{ (x,a ), (x,b) \}$}
		\ElsIf{$p> 0$}
		\label{alg:line:alternative_algorithm}
		\State $P\gets \{(x,y) \in T': y\in \T'\ \forall \T'\in T' \land x\notin \pi_2(C) \}$
		\If {$|P|>8k^\star$}
		\State\Return $\emptyset$
		\EndIf
		\For{$(x,y)\in P$}
		\State $C''\gets C\setminus \{(x,y)\}$
		\If{$|\{(x,z)\in C \}|=1$ }
		\State $C''\gets C''\setminus \{(x,z)\in C \}$
		\EndIf

		\State $R\gets R\ \cup\ \{(x,y)|r:r\in $ \Call{SemiTemporalCherryPicking}{$T'((x,y))$, $k'-1$, $k^\star$, $p-1$, $C''$}$\}$.\label{line:non_temporal_return}
		\EndFor
		\EndIf
		\State\Return $\{f|r: r \in R\}$
		\EndProcedure
	\end{algorithmic}
\end{algorithm}

\begin{algorithm}
	\begin{algorithmic}
		\Procedure{Pick}{$T', k', C'$}
		\State $(T^{(1)},k_1,C_1)\gets (T',k',C')$
		\State $p^{(1)}\gets ()$
		\State $i\gets 1$
		\While{$\exists x_i\in H(T^{(i)}): \left(w_{T^{(i)}}(x_i)=0\lor \{(x_i,n) : n \in N_{T^{(i)}}(x_i) \} \subseteq C_i \right) \land (\forall y\in N_{T^{(i)}}\forall \T\in T:y\in \T)$}
		\State $(n_1,\ldots, n_t)\gets N_T(x_i)$
		\State $p^{(i+1)}\gets p^{(i)}|((x_i, n_1),\ldots, (x_i,n_t))$
		\State $k_{i+1}\gets k_{i}-w_{T^{(i)}}(x_i)$
		\State $T^{(i+1)}\gets  T^{(i)}$
		\State $C_{i+1}\gets \{c\in C_{i}: x \notin \{c_1, c_2\}\}$
		\State $i\gets i+1$
		\EndWhile
		\State \Return $T^{(i)},k_{i},C_{i},p_{i}$
		\EndProcedure
	\end{algorithmic}
	\label{alg:non_temporal_pick}
	\caption{}
\end{algorithm}

\begin{lem}[Correctness of \texttt{Pick}] Suppose \texttt{Pick}$(T',k', C')$ in \cref{alg:non_temporal_pick} returns $(T,k,C,p)$.
	Then a tree-child sequence $s$ of weight at most $k$ for $T$ that satisfies $C$ exists if and only if a tree-child sequence $s'$ of weight at most $k'$ for $T'$ that satisfies $C'$ exists.
	In this case $p|s$ is a tree-child sequence for $T'$ of weight at most $k'$ and satisfying $C'$.
	\label{lem:tc:correctness_pick}
\end{lem}
The proof for this lemma is the same as for \cref{lem:correctness_pick}, but uses \cref{lem:tc:remove_immediately} instead of \cref{lem:remove_immediately}.
The following lemma was proven in \cite[Lemma 11]{van_iersel_practical_2019}.
\begin{lem}
	Let $s^a|s^b$ be a tree-child sequence for $T$ with weight $k$.
	If $T(s^a)$ contains no trivial cherries, then the number of unique cherries is at most $4k$.
	\label{lem:tc:unique_cherries}
\end{lem}

\begin{lem}
	If $((x_1,y_1),\ldots, (x_2,y_2),(x_{r+1},-), ,(x_{t},-))$ is a full tree child-sequence of minimal length for $T$ satisfying $C$ and  $H(T)\setminus \pi_2(C)=\emptyset$, then $(x_1,y_1)$ is a non-temporal element.
	\label{lem:first_non_temporal}
\end{lem}
\begin{proof}
	First observe that $x_1\notin \pi_2(C)$ because the sequence satisfies $C$.
	Suppose $(x_1,y_1)$ is a temporal element.
	This implies that there is an $i$ such that for all $j<i$ we have $x_j=x_1$ and $x_k\neq x_1$ for all $k\geq i$.
	This implies that for every $\T \in T$ there is a $j<i$ such that $x_1$ is not in $\T((x_j,y_j))$.
	Consequently $(x_j,y_j)$ is a cherry in $\T$.
	Because this holds for every tree $\T \in T$ we must have $H(T)\setminus \pi_2(C)$, contradicting the assumption that $H(T)\setminus \pi_2(C)=\emptyset$.
\end{proof}
\subsection{The algorithm}

\markj{
We now present our algorithm for \textsc{Semi-temporal hybridization}. As with \textsc{Tree-child hybridization}, we split the algorithm into two parts: \texttt{SemiTemporalCherryPicking}(\cref{alg:better_temporal2}) is the main recursive procedure, and \texttt{Pick}(\cref{alg:non_temporal_pick}) is the auxiliary procedure.}

\markj{The key idea is that we try to follow the procedure for temporal sequences as much as possible. \cref{alg:better_temporal2} only differs from \cref{alg:better_temporal} in the case where neither of the recursion conditions of \cref{alg:better_temporal} apply, but there are still cherries to be processed.
In this case, we can show that there are no trivial cherries, and hence \cref{lem:tc:unique_cherries} applies. 
Then we may assume there are at most $4k^*$ unique cherries, where $k^*$ is the original value of $k$ that we started with. In this case, we branch on adding $(x,y)$ or $(y,x)$ to the sequence, for any $x$ and $y$ that form a cherry. Any such pair will necessarily be a non-temporal element, and so we decrease $p$ by $1$ in this case.
A full proof of the following lemma is given in the appendix.}

\FloatBarrier

\newcommand{\lemTCProcedureReturnsSequenceText}{
	Let $s^\star$ be a tree-child sequence prefix, $T^\star$ a set of trees with the same leaves and define $T:=T^\star(s)$.
	Suppose $k,p\in \mathbf{N}$ and $C\in \mathcal{L}(T)^2$.
	When a generalized cherry picking sequence $s$ exists that satisfies $C$ and such that $s^ \star|s$ is a tree-child sequence for $T^\star$ with $w_{T^\star}(s^\star|s)\leq k^\star$ and $d(s)\leq p$ exists, \texttt{SemiTemporalCherryPicking}$(T, k, k^\star, p, C)$ from  \Cref{alg:better_temporal2} returns a non-empty set. }
\newtheorem*{lemTCProcedureReturnsSequence}{Lemma~\ref{lem:tc:procedure_returns_sequence}}
%First appearance
\begin{lem}
	\label{lem:tc:procedure_returns_sequence}
   \lemTCProcedureReturnsSequenceText
\end{lem}

\begin{lem}
	\label{lem:tc:returned_sequences_satisfy_demands}
	Let $s^\star$ be a tree-child sequence prefix, $T^\star$ a set of trees with the same leaves and define $T:=T^\star(s)$.
	Suppose $k,p\in \mathbf{N}$ and $C\in \mathcal{L}(T)^2$.
	If \leo{$S$ is returned by} a call to \texttt{SemiTemporalCherryPicking}$(T, k, k^\star, p, C)$, then for every $s\in S$, the sequence $s'=s^\star | s$ is a tree-child sequence for $T^\star$ with $d(s)\leq p$ and $w(s)\leq k$. %\todo[color=green!40]{Just write: similar to proof of Lemma 3.12, using \cref{lem:tc:correctness_pick}.}
\end{lem}
The proof of this lemma is similar to the proof of \cref{lem:procedure_returns_sequence} using \cref{lem:tc:correctness_pick}.

\begin{lem}
	\Cref{alg:better_temporal2} has a running time of $O(5^k\cdot (8k)^p\cdot k \cdot n\cdot m)$. \label{lem:tc:running_time_better_temporal} %\todo[color=green!40]{do we need the star?}
%	\todo{Similar to the temporal case, I am not convinced that the \textt{Pick} procedure can be run in time $O(nm)$ time (or at least we need more argument to avoid $O(n^2m)$)}
\end{lem}
\begin{proof}

	This can be proven by combining the proofs from \cref{lem:running_time_better_temporal} and \cite[Lemma 11]{van_iersel_practical_2019}.

	%\todo{}
\end{proof}

\begin{thm}
	\texttt{SemiTemporalCherryPicking}$(T, k, k, p, \emptyset)$ from  \Cref{alg:better_temporal2} returns a cherry picking sequence of weight at most $k$ if and only if such a sequence exists.
	The algorithm terminates in $O(5^k\cdot (8k)^p\cdot k \cdot n\cdot m)$ time.
\end{thm}
\begin{proof}
	This follows directly from \cref{lem:tc:returned_sequences_satisfy_demands}, \cref{lem:tc:procedure_returns_sequence} and \cref{lem:tc:running_time_better_temporal}.
\end{proof}

% ---------------- NONBINARY ------------------

\section{Constructing temporal networks from two non-binary trees}
\label{sec:non_binary_trees}
The algorithms described in the previous sections only work when all input trees are binary.
In this section we introduce the first algorithm for constructing a minimum temporal hybridization number for a set of two non-binary input trees.
The algorithm is based on \cite{piovesan_simple_2013} and has time complexity $O(6^kk!\cdot k
	\cdot n^2)$. 
	%\todo{check running time}

We say that a binary tree $\T'$ is a refinement of a non-binary tree $\T$ when $\T$ can be obtained from $\T'$ by contracting some of the edges.
Now we say that a network $\N$ displays a non-binary tree $\T$ if there exists a binary refinement $\T'$ of $\T$ such that both $\N$ displays $T'$.
Now the hybridization number $h_t(T)$ can be defined for a set of non-binary trees $T$ like in the binary case.

\begin{defn}
	A set $S\subseteq N_T(x)$ is a \emph{neighbor cover} for $x$ in $T$ if $S\cap N_\T(x) \neq \emptyset$ for all $\T\in T$.
\end{defn}

\begin{defn}
	For a set of non-binary trees $T$, define $w_T(x)$ as the minimum size of a neighbor cover of $x$ in $T$ minus one.
\end{defn}
Note that computing the minimum size of a neighbor cover is a NP-hard problem itself.
However if $|T|$ is constant the problem can be solved in polynomial time.
Note that for binary trees this definition is equivalent to the definition given in \cref{def:weight}.

Next \cref{def:ht} is generalized to non-binary trees.
\begin{defn}
	For a set of non-binary trees $T$ on the same taxa define $H(T) =\{x\in\Leav(T): \forall \T \in T \text{ }N_\T(x)\neq \emptyset  \}$.
\end{defn}

The non-binary analogue of \cref{def:cps} is given by the following lemma.
\begin{defn}
	For a set of non-binary trees $T$ with $n=\Leav(T)$, let $s=(s_1,\ldots, s_{n-1})$ be a sequence of leaves.
	Let $T_0=T$ and $T_i=T_{i-1}\setminus \{s_1, \ldots, s_i\}$.
	The sequence $s$ is a \emph{cherry picking sequence} if for all $i$, $s_i\in H(T\setminus\{s_1, \ldots, s_{i-1}\})$.
	Define the \emph{weight} of the sequence as $w_T(s)=\sum_{i=1}^{n-1} w_{T_{i-1}}(s_i)$.
\end{defn}

\begin{lem}
	A temporal network $\N$ that displays a set of nonbinary trees $T$ with reticulation number $r(\N)=k$ exists if and only if a cherry picking sequence of weight at most $k$ exists.
\end{lem}
\begin{proof}
	Note that this is a generalization of \cref{lem:exists_cherry_sequence} to the case of non-binary input trees and the proof is essentially the same.
	A cherry picking sequence with weight $k$ can be constructed from a temporal network with reticulation number $k$ in the same way as in the proof of \cref{lem:exists_cherry_sequence}.

	The construction of a temporal network $\N$ from a cherry picking $s$ is also very similar to the binary case: for cherry picking sequence $s_1,\ldots, s_t$, define  $\N_{t+1}$ to be the network, only consisting of a root, the only leaf of $T\setminus \{s_1, \ldots, s_t \}$ and an edge between the two.
	For each $i$ let $S_i$ be a minimal neighbor cover of $s_i$ in $T\setminus \{s_{1}, \ldots, s_{i-1} \}$.
	Now obtain $\N_{i}$ from $\N_{i+1}$ by adding node $s_i$, subdividing $(p_x,x)$ for every $x\in S_i$ with node $q_x$ and adding an edge $(q_x,s_i)$ and finally suppressing all nodes with in- and out-degree one.
	It can be shown that $r(\N)=w_T(s)$.
\end{proof}

\begin{lem}
	If $s$ is a cherry picking sequence for $T$ and for $x\in H(T)$ we have $w_T(x)=0$ then there is a cherry picking sequence $s'$ for $T$ with $w_T(s')=w_T(s)$ and $s'_1=x$.
	\label{lem:nonbin:remove_trivial}
\end{lem}
\begin{proof}
	We have $N_T(x)=\{y\}$.
	Now let $z$ be the element of $\{x,y\}$ that appears in $s$ first with $s_i=z$.
	Now $s'=(s_i,s_1,\ldots, s_{i-1},s_{i+1},\ldots)$ is a cherry picking sequence for $T$ with $w_T(s')=w_T(s)$.
	If $z=x$, then this proves the lemma.
	Otherwise we note that by swapping $x$ and $y$ in $T$, the trees stay the same.
	So we can also swap $x$ and $y$ in $s'$ without affecting the weight.
	Now $s'=x$, which proves the lemma.
\end{proof}

The algorithm relies on some theory from \cite{piovesan_simple_2013}, that we will introduce first.

For a vertex $v$ of $\T$ we say that all vertices reachable by $v$ form a pendant subtree.
For a pendant subtree $S$ we define $\mathcal{L}(S)$ set of the leaves of $S$.
Now we define
\begin{align*}
	Cl(\mathcal{T}) = \{\mathcal{L}(S): S \text{ is a pendant subtree of } \mathcal{T} \}\text{.
	}
\end{align*}
We call this the set of \emph{clusters} of $\T$.
Then we define $Cl(T)=\bigcup_{\T \in T}Cl(\T)$.
Call a cluster $C$ with $|C|=1$ \emph{trivial}.
Now we call a nontrivial cluster $C\in Cl(T)$ a  \emph{minimal} cluster if there is no $C'\in Cl(T)$ with $C'$ nontrivial and $C'\subsetneq C$.

In a cherry picking sequence $s$ we say that at index $i$ the cherry $(s_i,y)$ is \emph{reduced} if there is a $\T\in T$ such that $N_{T\setminus \{s_1,\ldots, s_{i-1}\}}(s_i)=\{y\}$.
\begin{lem}
	Let $T$ be a set of trees with $|T|=2$ such that $T$ contains no leaf $x$ with $w_T(x)=0$.
	Let $s$ be a cherry picking sequence for $T$.
	Then there is a minimal cluster $C$ in $T$ and a cherry picking sequence $s'=(s'_1,\ldots)$ for $T$ with $s'_i\in C$ for $i =1,\ldots, |C| -1$ and $w_T(s')\leq w_T(s)$.
	\label{lem:can_pick_min_cluster}
\end{lem}
\begin{proof}
	Let $p$ be the first index that a cherry is reduced in $s$.
	Let $(a,b)$ be one of the cherries that is reduced at index $p$.
	Now there will be a cherry in $T$ that contains both $a$ and $b$.
	Let $C$ be one of the minimum clusters that is contained in this cherry.
	Let $x$ be the element of $C$ that occurs last in $s$.
	Now let $c_1,\ldots, c_t$ be the elements from $C\setminus \{x\}$ ordered by their index in $s$.
	Now we claim that for any permutation $\sigma$ of $[t]$ we have $s'=(c_{\sigma(1)},\ldots,c_{\sigma(t)})|(s\setminus (C\setminus \{x\}))$ is a cherry picking sequence for $T$ and $w_T(s')\leq w_T(s)$.

	Let $i$ be the index of the last element of $C\setminus \{x\}$ in $s$.
	Suppose that $s'$ is not a CPS for $T$.
	Let $j$ be the smallest index for which $s'_j\notin H(T\setminus \{s'_1, \ldots, s'_{j-1} \})$.

	Let $\T \in T$ be such that $s'_j$ is not in a cherry in $T\setminus \{s'_1, \ldots, s'_{j-1}\}$.
	Choose $k$ such that $s_k=s'_j$.
	Now there are three cases:
	\begin{itemize}
		\item Suppose $j> i$, then $k=j$ and $\{s_1,\ldots, s_k \}=\{s'_1,\ldots, s'_j \}$.
		      This implies that $s'_j\in H(T\setminus \{s'_1,\ldots, s'_j\})$, which contradicts our assumption.
		\item Otherwise, suppose $s'_j\in \{c_1,\ldots, c_t\}$.
		      Then $j\leq t$.
		      Now $s_k$ has to be in a cherry in $\T\setminus \{s_1,\ldots, s_{k-1}\}$.
		      Because no cherries are reduced before index $i$ in $s$ this means that $s'_j$ is in a cherry in $\T$.
		      Because no cherries are reduced in $s'$ before index $t$, this implies that the same cherry is still in $\T\setminus \{s'_1,\ldots, s'_{j-1} \}$, which contradicts our assumption.
		\item Otherwise we must have $j\leq i$.
		      Because no cherries are reduced before index $i$ in $s$ this means that $s'_j$ is in a cherry $Q$ in $\T$.
		      If this cherry contains a leaf $y$ with $s'_w=y$ for $w>j$, then $s'_j$ is still in a cherry in $\T\setminus \{s'_1,\ldots, s'_{j-1} \}$, contradicting our assumption, so this can not be true.
		      However, that implies that the neighbors of $s_k$ in  $\T \setminus \{s_1, \ldots, s_{k-1} \}$ are all elements of $\{c_1, \ldots, c_t \}$.
		      Let $v$ be the second largest number such that $c_v$ is one of these neighbors.
		      Let $q$ be the index of $c_v$ in $s$.
		      Now cherry $Q$ will be reduced by $s$ at index $\max (q, j)< i$, which contradicts the fact that $C$ is contained in a cherry of $T$ that is reduced first by $s$.
	\end{itemize}

	Now to prove that $w_T(s')\leq w_T(s)$, we will prove that for $s_j=s'_k$ we have
	\begin{align*}
		w_{T\setminus \{s_1, \ldots, s_{j-1} \}}(s_j)\geq w_{T\setminus \{s'_1, \ldots, s'_{k-1} \}}(s'_k)\text{.
		}
	\end{align*}
	Note that for $j\geq i$ this is trivial, so assume $j<i$.
	If $s_j\in C\setminus \{x\}$, then $w_{T\setminus \{s_1, \ldots, s_{j-1} \}}(s_j)\geq w_{T}(s_j)$ because no cherries are reduced before $i$, which implies that no new elements added to cherries before $i$.
	For the same reason we must have $s_j\in H(T)$.
	Because there are no $x\in H(T)$ with $w_T(x)=0$ we must have $w_{T}(s_j)=1$.
	So $w_{T\setminus \{s'_1, \ldots, s'_{k-1} \}}(s'_k)\leq  w_{T\setminus \{s_1, \ldots, s_{j-1} \}}(s_j)=1$.
\end{proof}

\subsection{Bounding the number of minimal clusters}
By \cref{lem:can_pick_min_cluster} in the construction of a cherry picking sequence we can restrict ourselves to only appending elements from minimal clusters.
We use the following theory from \cite{piovesan_simple_2013} to bound the number of minimal clusters.
\begin{defn}
	Define the relation $x\xrightarrow[]{T} y$ for leaves $x$ and $y$ of $T$ if every nontrivial cluster $C\in Cl(T)$ also contains $y$.
\end{defn}
\begin{obs}[{\cite[Observation 2]{piovesan_simple_2013}}]
	The relation $\xrightarrow[]{T}$ defines a partial ordering on $\Leav(T)$.
\end{obs}

Now call $x\in \Leav(T)$ a \emph{terminal} if there is no $y\neq x$ with $x\xrightarrow[]{T} y$.
Now we will first show that all minimal clusters contain a terminal.
Then a bound on the number of terminals gives a bound on the number of minimal clusters.

\begin{lem}
	Every minimal cluster contains a terminal.
\end{lem}
\begin{proof}
	Let $C$ be a minimal cluster of $T$.
	Let $x$ be an element of $C$ that is maximal in $C$ with respect to the partial ordering `$\xrightarrow[]{T}$' (if we say that $x\xrightarrow[]{T}y$ means that $y$ is `greater than or equal to' $y$).
	Now suppose that $x$ is not a terminal.
	Then there is an $y$ such that $x\xrightarrow[]{T}y$.
	However then $y\in C$, but this contradicts the fact that $x$ is a maximal element in $C$ with respect to `$\xrightarrow[]{T}$'.
	Because this is a contradiction, $x$ has to be a terminal.
\end{proof}

\begin{lem}
	\label{lem:max3kterminals}
	Let $T$ be a set of trees with $h_t(T)\geq1$ containing no zero-weight leaves.
	Let $\N$ be a network that displays $T$.
	Then $T$ contains at most $2r(\N)$ terminals that are not directly below a reticulation node.
\end{lem}
\begin{proof}
	We reformulate the proof from  \cite[Lemma 3]{piovesan_simple_2013}.
	We use the fact that for each terminal one of the following conditions holds: the parent $p_x$ of $x$ in $N$ is a reticulation (condition 1) or a reticulation is reachable in a directed tree-path from the parent $p_x$ of $x$ (condition 2).
	This is always true because if neither of the conditions holds, because otherwise another leaf $y$ is reachable from $p_x$, implying that $x\xrightarrow[]{T}y$, which contradicts that $x$ is a terminal.

	Let $R$ be the set of reticulation nodes in $\N$ and let $W$ be the set of terminals in $T$ that are not directly beneath a reticulation.
	We describe a mapping $F:W\to R$ such that each reticulation $r$ is mapped to at most $d^-(r)$ times.
	Note that for each $x\in W$ condition 2 holds.
	For these elements let $F(x)=y$ where $y$ is a reticulation reachable from $p(x)$ by a tree-path.
	Note that there can not be a path from $p(x)$ to  $y$ containing only tree arcs when $x\neq y$ are both in $H(T)$ because then $x\to y$ which contradicts that $x$ is a terminal.
	It follows that each reticulation $r$ can be mapped to at most $d^-(r)$ times: at most once incoming edge.
	Then for the set of terminals $\Omega$ we have $|\Omega| \leq \sum_{r\in R} d^-(r) \leq \sum_{r\in R} (1+(d^-(r)-1)) \leq |R| + k \leq 2k$.
\end{proof}

\begin{lem}
	\label{lem:max2kterminals}
	Let $T$ be a set of nonbinary trees such that $h_t(T)\geq1$.
	Then any set $S$ of terminals in $T$ with $|S|\geq 2h_t(T)+1$ contains at least one element $x\in H(T)$ such that $s$ is a cherry picking sequence for $T$ with $w_T(s)=h_t(T)$ and $s_1=x$.
\end{lem}
\begin{proof}
	Let $\mathcal{N}$ be a temporal network that displays $T$ such that $r(\N)=h_t(T)$ with corresponding cherry picking sequence $s$.
	From the \cref{lem:max3kterminals} it follows that at most $r(\N)$ terminals exist in $T$ that are not directly below a reticulation.
	So there is an $x\in S$ that is directly below a reticulation.

	Now let $T'$ be the set of all binary trees displayed by $\N$.
	Note that $s$ is a cherry picking sequence for $T'$.
	Let $i$ be such that $s_i=x$.
	Because $x$ is directly below a reticulation in $\N$, we have $s_j\notin N_{T'}(x)$, which implies by \cref{lem:twoconditions} that $s'=(s_i,s_1,\ldots, s_{i-1},s_{ i+1},\ldots)$ is a cherry picking sequence for $T'$ with $w_{T'}(s')=w_{T'}(s)=r(\N)=h_t(T)$.
	Now $w_{T}(s')\leq w_{T'}(s')=h_t(T)$, so $w_{T}(s')=h_t(T)$.
\end{proof}

\begin{algorithm}
	\caption{\label{alg:non_binary_new}}

	\begin{algorithmic}[1]
		\Procedure{CherryPicking}{$T,k$}
		\State $s\gets ()$
		\While{$\exists x\in H(T): w_T(x)=0$}
		\State $T\gets T\setminus \{x\}$ \label{line:remove_trivial}
		\State $s\gets s|(x)$
		\EndWhile\\
		\If{$|\Leav(T)|=1$}
		\State\Return $\{s\}$
		\ElsIf{$k=0$}
		\State \Return $\emptyset$
		\EndIf\\
		\State $S\gets $ set of terminals in $T$
		\If{$|S|>2k$}
		\State $S'\gets $subset of $S$ of size $2k+1$
		\For{$x\in S'\cap H(T)$}
		\State $R\gets R \cup \{ (x)\ |\ x : x \in $ \Call{CherryPicking}{$T\setminus \{x \}, k-1$ }  $\}$
		\EndFor
		\Else
		\For{$q\in S$}
		\State $D\gets$ set of minimum clusters that contain $q$
		\If{$\exists y,z: D=\{\{q,y\},\{q,z\}\}$\label{line:if_split_3} }
		\For{$x\in \{q,y,z\} \cap H(T)$}
		\State $R\gets R \cup \{ (y)\ |\ x : x \in $ \Call{CherryPicking}{$T\setminus \{y\}, k-1$}$\}$
		\EndFor
		\Else
		\For{$C\in D$}
		\For{$x\in C: C\setminus \{x\} \subseteq H(T)$}
		\State $(c_1,\ldots, c_t)\gets C\setminus \{x\}$
		\State $R\gets R \cup \{ (c_1,\ldots, c_t)\ |\ x : x \in $ \Call{CherryPicking}{$T\setminus \{c_1, \ldots, c_t \}, k-t$}$\}$
		\EndFor
		\EndFor
		\EndIf
		\EndFor
		\EndIf
		\State\Return $\{s|x:x\in R\}$
		\EndProcedure
	\end{algorithmic}
\end{algorithm}
\subsection{Run-time analysis}
\begin{lem}
	The running time of $\texttt{CherryPicking}(T,k)$ from \cref{alg:non_binary_new} is $O(6^kk!
		\cdot k\cdot n^2)$ if $T$ is a set consisting of two nonbinary trees.
\end{lem}
\begin{proof}
	%\todo{Similar to other algorithms, I don't think the non-recursive part takes $O(n^2)$ time - the while loop on line 3 takes $O(m)$ to check $x \in H(T)$ for each of $O(n)$ leaves and may need to be run $O(n)$ times.}
	Let $f(n)$ be an upper bound for the running time of the non-recursive part of the function.
	We claim that the maximum running time $t(n,k)$ for running the algorithm on trees with $n$ leaves and parameter $k$ is bounded by $6^{k}k!
		kf(n)$.

	For $k=0$ it is clear that this claim holds.
	Now we will prove that it holds for any call, by assuming that the bound holds for all subcalls.

	If $|S|>2k$, then the algorithm branches into $2k+1$ subcalls.
	The total running time can then be bounded by
	\begin{align*}
		(2k+1)t(n,k-1) + f(n) & \leq (2k+1)6^{k-1}(k-1)!
		(k-1)f(n)+ f(n)
		\\&\leq 6^{k}(k)!(k)f(n)\text{.}
	\end{align*}

	%	Now consider the case that $|S|\leq 2k$. Now for each $q\in S$ there are two possibilities.
	If the condition of the if-statement on \cref{line:if_split_3} is true, then for that $q$ the functions does $3$ subcalls with $k$ reduced by one.
	So the recursive part of the total running time for this $q$ is bounded by
	\begin{align*}
		3T(k-1)\leq 6^{k-1}(k-1)!
		(k-1)f(n) = 3^{k}2^{k-1}(k-1)!(k-1)f(n) \text{.}
	\end{align*}
	If the condition on \cref{line:if_split_3} holds then there is at most one $d\in D$ with $|d|\leq 2$.
	Using this information we can bound the total running time of the subcalls that are done for $q$ in the else clause by
	\begin{align}
		 & \sum_{d\in D}|d|t(k-|d|+1)\leq \sum_{d\in D}|d|6^{k-|d|+1}(k-|d|+1)!
		(k-|d|+1)f(n)\nonumber                                                                 \\
		 & \leq \sum_{d\in D}|d|6^{k-|d|+1}(k-|d|+1)!(k-|d|+1)f(n)\nonumber                    \\
		 & \leq (k-1)!(k-1)f(n) \sum_{d\in D}|d|6^{k-|d|+1}                                    \\
		 & \leq (k-1)!(k-1)f(n)(2\cdot 6^{k-1} + 3\cdot 6^{k-2})\label{eq:decreasing_function} \\
		 & = (k-1)!(k-1)f(n)2^{k-1}(9\cdot 3^{k-2})\nonumber                                   \\
		 & =(k-1)!(k-1)f(n)2^{k-1}3^{k}
		\text{.}
	\end{align}
	Note that \cref{eq:decreasing_function} follows from  the fact that $x\mapsto x 6^{k-x+1}$ is a decreasing function for $x\in [1,\infty)$.
	So for each $q$ the running time of the subcalls is bounded by $(k-1)!
		(k-1)f(n)2^{k-1}3^{k}$.
	Now the total running time is bounded by 

	\begin{align}
		 & \phantom{\leq .
		} f(n) + (k-1)!(k-1)f(n)2^{k-1}3^{k}|S|         \\
		 & \leq  f(n) +   (k-1)!(k-1)f(n)2^{k-1}3^{k}2k \\
		 & \leq f(n) +   k!(k-1)f(n)6^{k}               \\
		 & \leq 6^{k}k!kf(n)
	\end{align}
	Because the non-recursive part of the function can be implemented to run in $O(n^2)$ time 
	the total running time of the function is $O(6^{k}k!
		\cdot k\cdot n^2)$.
\end{proof}

\begin{lem}
	Let $T$ be a set of non-binary trees.
	If $h_t(T)\leq k$, then \texttt{CherryPicking}$(T,k)$ from \cref{alg:non_binary_new} returns a cherry picking sequence for $T$ of weight at most $k$.
\end{lem}
\begin{proof}
	First we will prove with induction on $k$ that if $h_t(T)\leq k$ then a sequence is returned.

	For $k=0$ it is true because if $h_t(T)=0$, as long as $\Leav(T)>1$ then $|H(T)|>0$ and all elements of $H(T)$ will have zero weight, so they are removed on \cref{line:remove_trivial}.
	After that $\Leav(T)=1$ so an empty sequence will be returned, which proves that the claim is true for $k=0$.

	Now assume that the claim holds for for $k<k'$ and assume that $h_t(T)\leq k'$.
	Now we will prove that a sequence is returned by \texttt{CherryPicking}($T,k$) in this case.
	After removing an element $x$ with weight zero on \cref{line:remove_trivial} we still have $h_t(T)\leq k'$ (\cref{lem:nonbin:remove_trivial}).
	If $|\Leav(T)|=1$, an empty sequence is returned.
	If this is not the case then $0<h_t(T)\leq k$, so the else if is not executed.

	If $|S|>2k$ then from \cref{lem:max2kterminals} it follows that for $S'\subseteq S$ with $|S'|=2k+1$ there is at least one $x\in S'$ such that $h_t(T\setminus \{x\})\leq k'-1$.
	Now from the induction hypothesis it follows that \texttt{CherryPicking}$(T\setminus \{x\},k')$ returns at least one sequence, which implies that $R$ is not empty.
	Because of that the main call will return at least one sequence, which proves that the claim holds for $k=k'$.

	The only thing left to prove is that every returned sequence is a cherry picking sequence for $T$.
	This follows from the fact that only elements from $H(T)$ are appended to $s$ and that $R$ consists of cherry picking sequences for $T\setminus \{s_1,\ldots,s_t \}$.
\end{proof}

% ---------------- RESULTS -----------------

\section{Experimental results}
\label{sec:implementation}
\FloatBarrier
We developed implementations of \cref{alg:better_temporal}, \cref{alg:non_binary_new} and \cref{alg:better_temporal2}, \leo{which are freely available}~\cite{sjb_implementation}.
To analyse the performance of the algorithms we made use of dataset generated in \cite{van_iersel_practical_2019} for experiments with an algorithm for construction of tree-child networks with a minimal hybridization number.
\subsection{\cref{alg:better_temporal}}
In \cref{fig:time_k_plot} the running time of \cref{alg:better_temporal} on the dataset from \cite{van_iersel_practical_2019} is shown.
The results are consistent with the bound on the running time that was proven in \cref{sec:algorithm}.
Also, the algorithm is able to compute solutions for relatively high values of $k$, indicating that the algorithm performs well in practice.
\begin{figure}\centering
	\includegraphics[scale=.8]{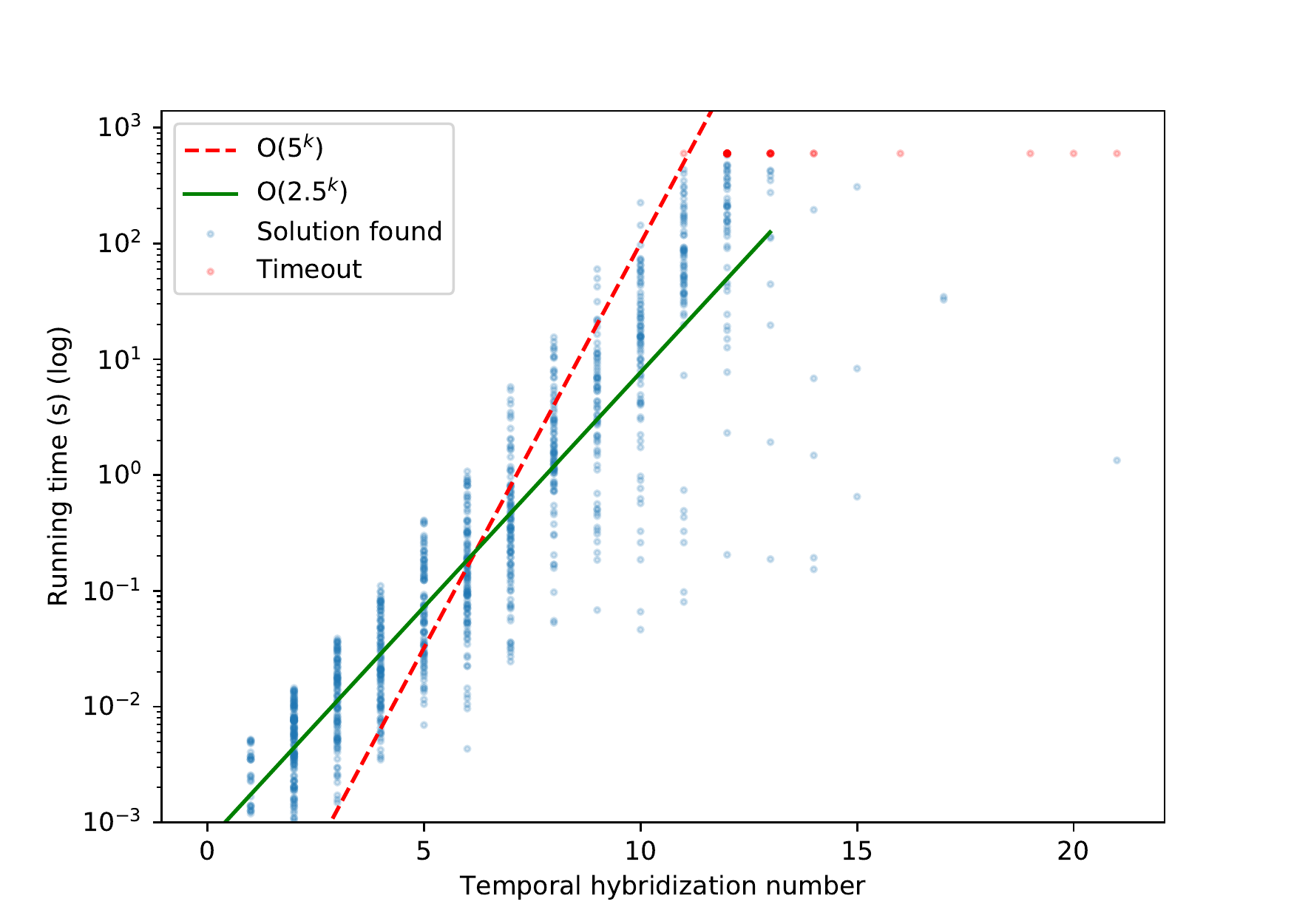}
	\caption{The running time of \cref{alg:better_temporal} on problem shown relative to the corresponding temporal hybridization number.
		A timeout of 10 minutes was used.
		Instances for which the algorithm timed out are shown in red at the value of $k$ where they timed out.
		On the log scale the exponential relation is clearly visible.
		However fitting an exponential function on the data yields a $O(2.5^k)$ function for temporal hybridization number $k$, while the worst-case bound that we proved is $O(5^k)$.
		\label{fig:time_k_plot}
	}
\end{figure}

The authors of \cite{van_iersel_practical_2019} also provide an implementation of their algorithm for tree-child algorithms.
The implementation contains several optimizations to improve the running time.
One of them is an operation called cluster reduction \cite{linz_cluster_2011}.
The implementation is also multi-threaded.
In \cref{fig:runtime_comparison} we provide a comparison of the running times of the tree-child algorithm with \cref{alg:better_temporal}.
In this comparison we let both implementations use a single thread, because our implementation of the algorithm for computing the hybridization number does not support multithreading.
The implementation could however be modified to solve different subproblems in different threads which will probably also result in a significant speed-up.
In \cref{alg:better_temporal} we see that the difference in time complexity between the $O((8k)^k)$ algorithm and the $O(5^k)$ algorithm is also observable in practice.
\begin{figure}\centering
	\includegraphics[scale=.8]{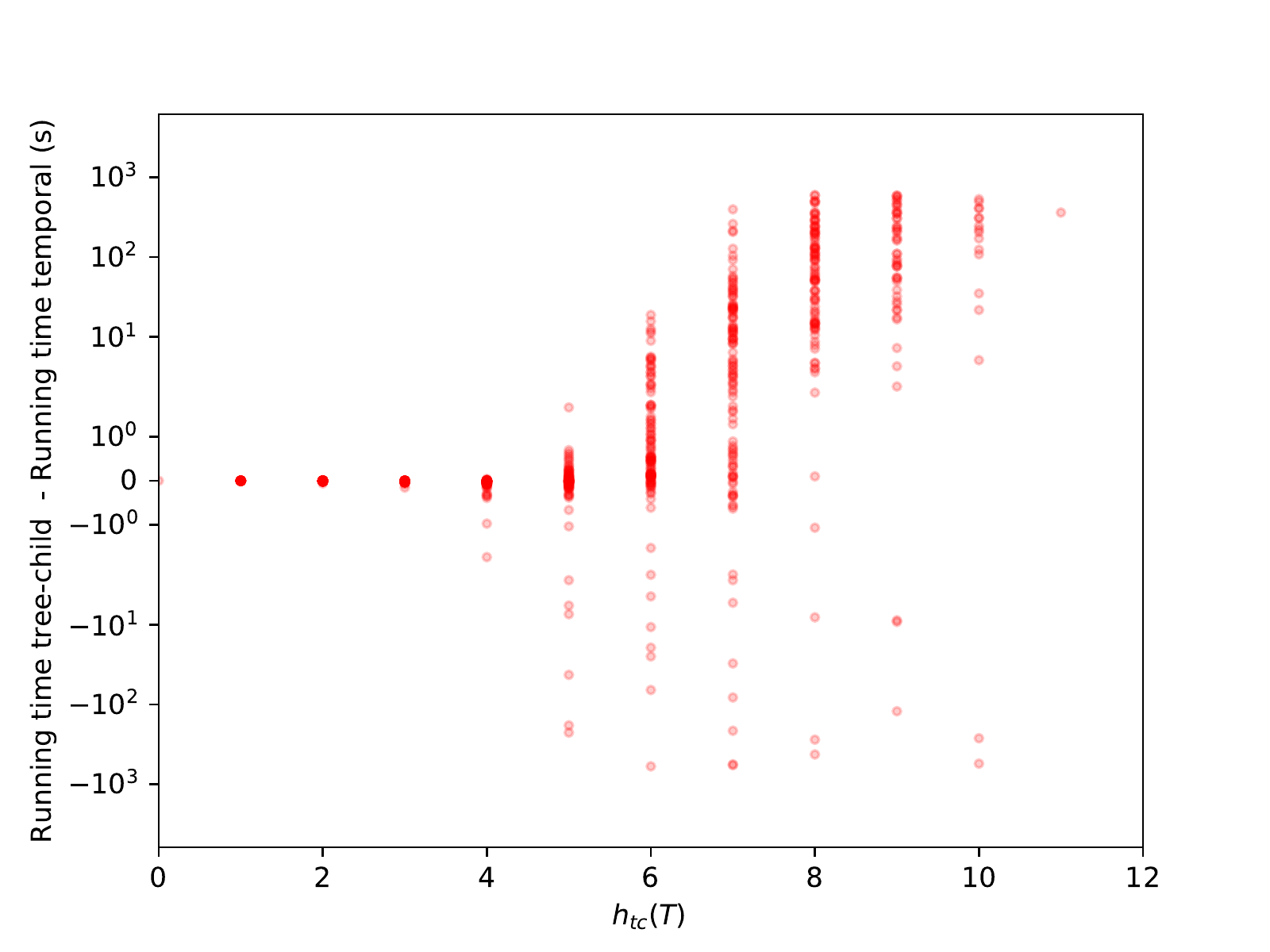}
	\caption{Difference between the running time of \cref{alg:better_temporal} and the algorithm for tree-child networks from \cite{van_iersel_practical_2019}.}
	\label{fig:runtime_comparison}
\end{figure}

\subsection{\cref{alg:non_binary_new}}
We used the software from \cite{van_iersel_practical_2019} to generate random binary problem instances and afterwards randomly contracted edges in the trees to obtain non-binary problem instances.
We used this dataset to test the running time of \cref{alg:non_binary_new}.
The results are shown in \cref{fig:time_k_plot_nonbinary}.
We see that the algorithm is usable in practice and has a reasonable running time.
\begin{figure}\centering
	\includegraphics[scale=.8]{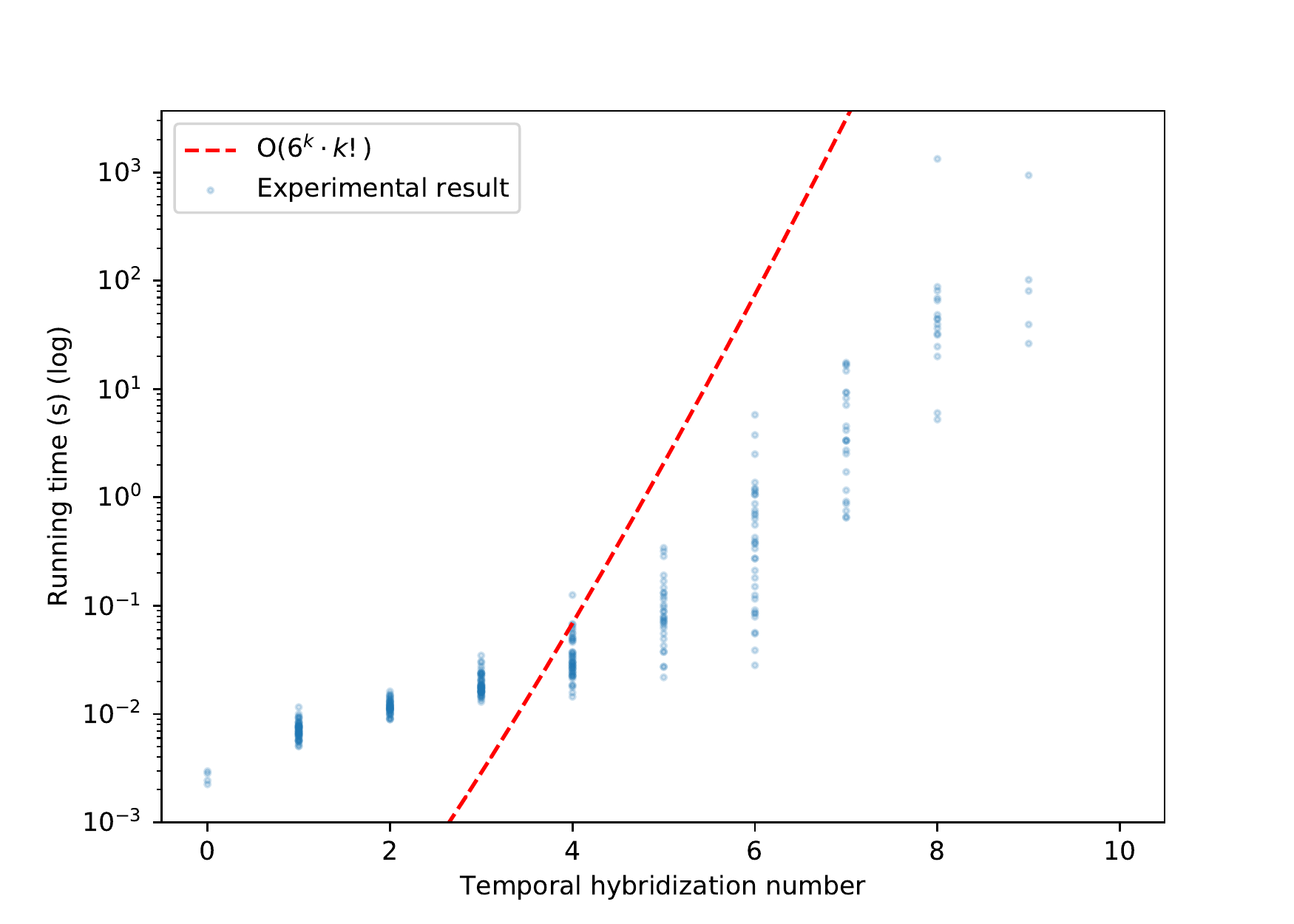}
	\caption{Running time of \cref{alg:non_binary_new} on a generated set of instances consisting of trees with average out-degree $2.5$ relative to the temporal hybridization number.
		A timeout of 10 minutes was used.
	}
	\label{fig:time_k_plot_nonbinary}
\end{figure}

\begin{figure}\centering
	\includegraphics[scale=.8]{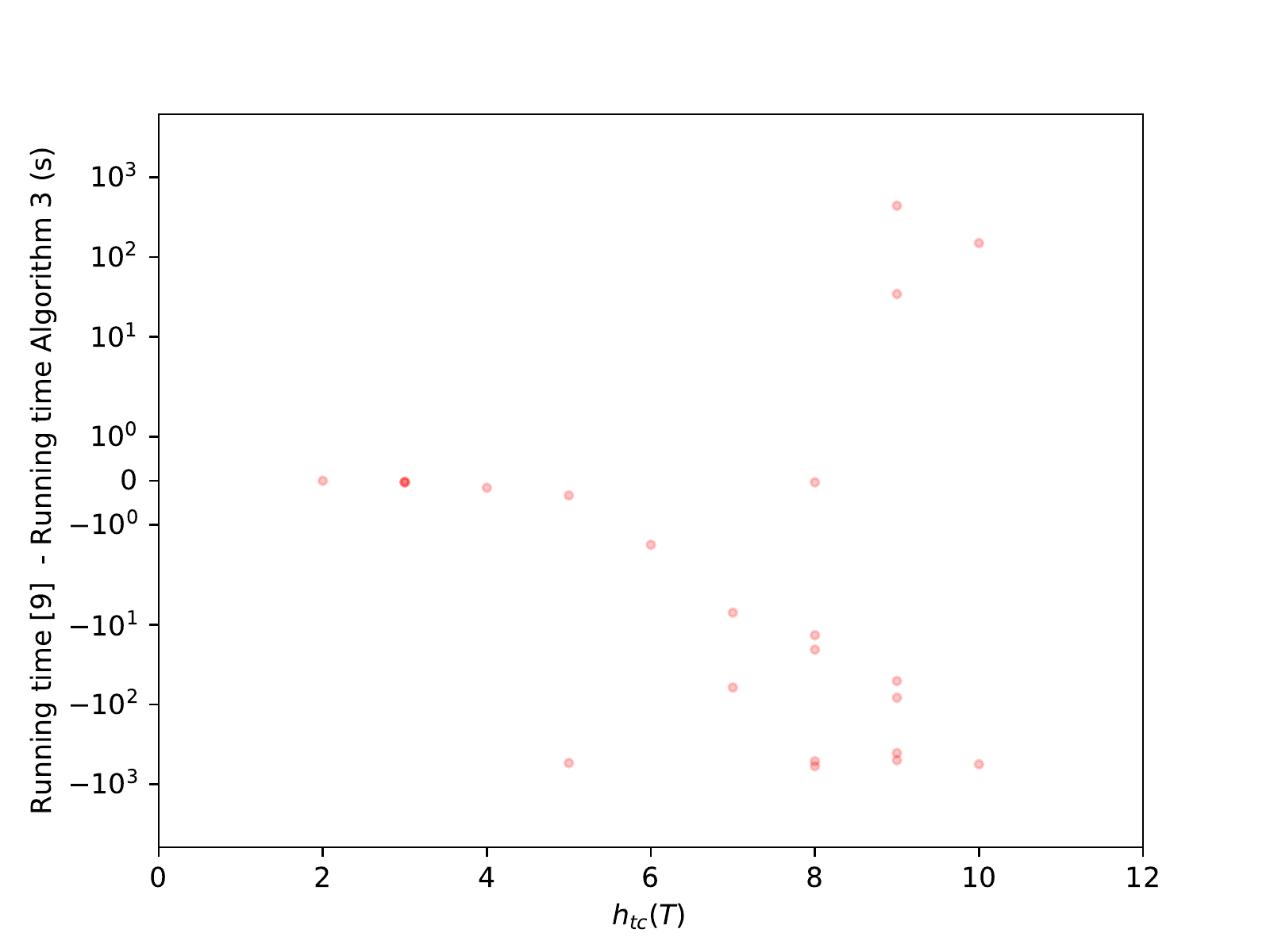}
	\caption{Difference between the running time of \cref{alg:better_temporal2} and the algorithm for constructing tree-child networks from \cite{van_iersel_practical_2019} on all non-temporal instances in the dataset from \cite{van_iersel_practical_2019}.}
	\label{fig:semi_temporal_vs_treechild_runtime}
\end{figure}
\FloatBarrier

\subsection{\cref{alg:better_temporal2}}
\cref{alg:better_temporal2} was tested on all non-temporal instances in the dataset from \cite{van_iersel_practical_2019}.
In \cref{fig:semi_temporal_vs_treechild_runtime} the running time of \cref{alg:better_temporal2} is compared to that of the algorithm from \cite{van_iersel_practical_2019}.
The data show that the algorithm from \cite{van_iersel_practical_2019} is often faster than \cref{alg:better_temporal2}. \leo{However, there also also instances for which \cref{alg:better_temporal2} is much faster. Hence, in practice it can be worthwile to run this algorithm on instances that cannot be solved by the algorithm from \cite{van_iersel_practical_2019} in a reasonable time.}
%This suggests that \cref{alg:better_temporal2} is mostly of theoretical value and is not very useful in practice.
It should also be noted that we only tested the algorithms on a relatively small dataset. 

%\todo{take out Fig.~\ref{fig:semi_temporal_vs_treechild_runtime}? I would keep it.}

% ----------- DISCUSSION -----------

\section{Discussion}
\label{sec:further_research}
\Cref{alg:better_temporal}, the algorithm for constructing minimum temporal hybridization networks, has a significantly better running time than the algorithms that were known before.
The results from the implementation show that the algorithm also works well in practice.
However this implementation could still be improved, for example by making use of parallelization.

While we also present an algorithm that finds optimal temporal networks for nonbinary trees, the running time of this algorithm is significantly worse and, moreover, it only works for pairs of trees.
An open question is whether this could be improved to a running time of $O(c^k\cdot poly(n))$ for some $c\in \mathbb{R}$, perhaps using techniques similar to our approach for binary trees.
Another important open problem is whether Temporal Hybridization is FPT for a set of more than two non-binary input trees.

In \cref{sec:non_temporal} a metric is provided to quantify how close a hybridization network is to being temporal.
However, other, possibly more biologically meaningful, metrics could also be used for this purpose.
An open problem is whether an FPT algorithm exists that solves the decision problem associated with these metrics.

\bibliographystyle{plain}
\bibliography{library}

\begin{thebibliography}{10}

\bibitem{bapteste_networks_2013}
Eric Bapteste, Leo van Iersel, Axel Janke, Scot Kelchner, Steven Kelk, James~O.
  McInerney, David~A. Morrison, Luay Nakhleh, Mike Steel, Leen Stougie, and
  James Whitfield.
\newblock Networks: expanding evolutionary thinking.
\newblock {\em Trends in Genetics}, 29(8):439--441, August 2013.

\bibitem{bordewich_computing_2007}
M.~Bordewich and C.~Semple.
\newblock Computing the {Hybridization} {Number} of {Two} {Phylogenetic}
  {Trees} {Is} {Fixed}-{Parameter} {Tractable}.
\newblock {\em IEEE/ACM Transactions on Computational Biology and
  Bioinformatics}, 4(3):458--466, July 2007.

\bibitem{bordewich_computing_2007-2}
Magnus Bordewich and Charles Semple.
\newblock Computing the minimum number of hybridization events for a consistent
  evolutionary history.
\newblock {\em Discrete Applied Mathematics}, 155(8):914--928, April 2007.

\bibitem{sjb_implementation}
Sander Borst.
\newblock Temporal hybridization number algorithm implementations, 2020.
\newblock http://github.com/mathcals/temporal\_hybridization\_number.

\bibitem{cardona_comparison_2007}
Gabriel Cardona, Francesc Rossello, and Gabriel Valiente.
\newblock Comparison of {Tree}-{Child} {Phylogenetic} {Networks}.
\newblock {\em arXiv:0708.3499 [cs, q-bio]}, August 2007.
\newblock arXiv: 0708.3499.

\bibitem{docker_deciding_2019}
Janosch Döcker, Leo van Iersel, Steven Kelk, and Simone Linz.
\newblock Deciding the existence of a cherry-picking sequence is hard on two
  trees.
\newblock {\em Discrete Applied Mathematics}, 260:131--143, May 2019.

\bibitem{humphries_cherry_2013}
Peter~J. Humphries, Simone Linz, and Charles Semple.
\newblock Cherry {Picking}: {A} {Characterization} of the {Temporal}
  {Hybridization} {Number} for a {Set} of {Phylogenies}.
\newblock {\em Bulletin of Mathematical Biology}, 75(10):1879--1890, October
  2013.

\bibitem{humphries_complexity_2013}
Peter~J. Humphries, Simone Linz, and Charles Semple.
\newblock On the complexity of computing the temporal hybridization number for
  two phylogenies.
\newblock {\em Discrete Applied Mathematics}, 161(7-8):871--880, May 2013.

\bibitem{linz_hybridization_2009}
S.~Linz and C.~Semple.
\newblock Hybridization in {Nonbinary} {Trees}.
\newblock {\em IEEE/ACM Transactions on Computational Biology and
  Bioinformatics}, 6(1):30--45, January 2009.

\bibitem{linz_cluster_2011}
Simone Linz and Charles Semple.
\newblock A {Cluster} {Reduction} for {Computing} the {Subtree} {Distance}
  {Between} {Phylogenies}.
\newblock {\em Annals of Combinatorics}, 15(3):465--484, September 2011.

\bibitem{linz_attaching_2019}
Simone Linz and Charles Semple.
\newblock Attaching leaves and picking cherries to characterise the
  hybridisation number for a set of phylogenies.
\newblock {\em Advances in Applied Mathematics}, 105:102--129, April 2019.

\bibitem{mallet_how_2016}
James Mallet, Nora Besansky, and Matthew~W. Hahn.
\newblock How reticulated are species?
\newblock {\em BioEssays}, 38(2):140--149, February 2016.

\bibitem{piovesan_simple_2013}
T.~Piovesan and S.~M. Kelk.
\newblock A {Simple} {Fixed} {Parameter} {Tractable} {Algorithm} for
  {Computing} the {Hybridization} {Number} of {Two} ({Not} {Necessarily}
  {Binary}) {Trees}.
\newblock {\em IEEE/ACM Transactions on Computational Biology and
  Bioinformatics}, 10(1):18--25, January 2013.

\bibitem{soucy_horizontal_2015}
Shannon~M. Soucy, Jinling Huang, and Johann~Peter Gogarten.
\newblock Horizontal gene transfer: building the web of life.
\newblock {\em Nature Reviews Genetics}, 16(8):472--482, August 2015.

\bibitem{van_iersel_practical_2019}
Leo van Iersel, Remie Janssen, Mark Jones, Yukihiro Murakami, and Norbert Zeh.
\newblock A {Practical} {Fixed}-{Parameter} {Algorithm} for {Constructing}
  {Tree}-{Child} {Networks} from {Multiple} {Binary} {Trees}.
\newblock {\em arXiv:1907.08474 [cs, math, q-bio]}, July 2019.
\newblock arXiv: 1907.08474.

\bibitem{van_iersel_kernelizations_2016-1}
Leo van Iersel, Steven Kelk, and Celine Scornavacca.
\newblock Kernelizations for the hybridization number problem on multiple
  nonbinary trees.
\newblock {\em Journal of Computer and System Sciences}, 82(6):1075--1089,
  September 2016.

\end{thebibliography}

\newpage

\appendix
\section{Omitted proofs}

%\begin{lem}
%	Let $s$ be a full tree-child sequence $s$ for $T$.
%	Then there exists a network $\N$ with semi-temporal labeling $t$ such that $r(\N)\leq w_T(s)$ and $d(\N,t)\leq d(s)$. \todo{move proof to appendix, include intuitive explanation here}
%\end{lem}
%Second appearance
\begin{lemSemiTemporalTequenceToNetwork}
   \lemSemiTemporalTequenceToNetworkText
\end{lemSemiTemporalTequenceToNetwork}
\begin{proof}
	This can be proven by constructing a tree-child network from the tree-child sequence as described in \cite[Proof of Theorem 2.2]{linz_attaching_2019}.
	We will show that a semi-temporal labeling satisfying our constraints exists for the resulting network.
	We will write
	\begin{align*}
		s=(x_1,y_1),\ldots, (x_r,y_r),(x_{r+1},-)
	\end{align*}
	Now we merge all consecutive elements $(x_{i},y_{i}), (x_{i+1},y_{i+1}),\ldots,(x_{i+j},y_{i+j})$ for which $x_{i}=x_{i+1}=\cdots =x_{i+j}$ into one element $(x_i, \{y_{i},y_{i+1},\ldots, y_{i+j} \})$ and call the resulting sequence $s'$.
	Call an element of this sequence \emph{temporal} if all corresponding elements in $s$ are temporal.
	Call it \emph{non-temporal} if all corresponding elements in $s$ are non-temporal.
	Observe that it can not happen that some of the corresponding elements are temporal while some are non-temporal.
	\begin{enumerate}
		\item Let $\N_{r+1}$ be the network consisting of root $\rho$, vertex $x_{r+1}$ and edge $(\rho, x_{r+1})$.
		      Set $i:=r$.
		      Set $t(\rho):=0$ and $t(x_{r+1}):=1$.
		\item If $i=0$, contract all edges with in- and out-degree $1$ in $\N_{1}$ and return the resulting network together with $t_1$.
		\item Set $t_{i}:=t_{i+1}$.
		\item For the element $s'_i=(x, Y)$ do the following:
		      \begin{enumerate}
			      \item \markj{If $s'_i$ is a temporal element then} $x\notin \mathcal{L}(\N_{{i+1}})$. %\todo[color=green!40]{Changed order of these cases}
			            In this case let $\N_i$ be the network obtained from $\N_{i+1}$ by adding vertex $x$, vertex $p_x$, edge $(p_x,x)$, subdividing edge $(p_y,y)$ by vertex $v_y$  and adding edge $(v_y,p_x)$ for all $y\in Y$.
			            Set $t_i(p_x):=-\infty$.
			            \label{item:case_edge_two}
			      \item \markj{Otherwise $s'_i$ is a non-temporal element} and $x\in \mathcal{L}(\N_{{i+1}})$.
			            In this case let $\N_i$ be the network obtained from $\N_{i+1}$ by subdividing  $(p_y, y)$ for all $y\in Y$ with a new vertex $v_y$ and adding the edge $(v_y,p_x)$.
			            \label{item:case_edge_one}
		      \end{enumerate}
		\item Set $\tau = \max \{\max_{y\in Y}t_{i}(p_y),t_i(p_x)-1\}$.
		      For all $y\in Y$ set $t_i(v_y):=\tau + 1$ and $t_i(y):= \tau + 2$.
		      If $s_i'$ is a temporal element set \markj{$t_i(p_x):=\tau+1$ and $t_i(x):=\tau+2$.} %\todo[color=green!40]{Fixed a few places where we had $t_i(x) = t_i(p_x)$ ($t_i(x)$ should be greater)}
		\item Decrease $i$ by one.
		      Go to step 2.
	\end{enumerate}
	Note that the construction of the network is equivalent to the one described in \cite[Proof of Theorem 2.2]{linz_attaching_2019}, where it is also proven that the resulting network is a tree-child network that is fully reduced by $s$.
	The only thing we have to prove is that $t$ is a semi-temporal labeling of $\N$ with $d(\N, t)\leq d(s)$.

	We will prove with induction on $i$ that $t_i$ is a semi-temporal labeling for $\N_i$.
	For $N_n$ it is clear that this is true.
	Consider an arbitrary edge $(u,v)$ in $\N_i$.
	If the edge was also in $\N_{i+1}$, then $t_{i+1}(u)=t_{i}(u)$ and $t_{i+1}(v)=t_{i}(v)$, so the edge satisfies the conditions for being semi-temporal.

	Now we will go through all newly introduced edges in $\N_i$ and show that they satisfy the conditions for being semi-temporal.
	\begin{itemize}
		\item In \cref{item:case_edge_two} edges $(p_y,v_y)$, $(v_y, y)$, and $(v_y, p_x)$ are created for all $y\in Y$ and $(p_x,x)$ is created.
		      Because $t_i(v_y) =\tau+1>t_i(p_y)$ and $t_i(y)=\tau +2 > \tau +1=t_i(v_y)$ the first two edges are semi-temporal.
		      Because in this case $s_i'$ is a temporal element \markj{$t_i(p_x):=\tau +1$ and $t_i(x):=\tau +2$} will be set in step 5, \markj{so $(p_x,x)$ is semi-temporal}.
		      Consequently \markj{$t_i(v_y)=\tau +1 =t_i(p_x)$}, so $(v_y, x)$ is also semi-temporal.
		\item In \cref{item:case_edge_one} edges $(p_y, v_y)$, $(v_y,y)$, $(v_y,p_x)$ for all $y\in Y$. 
		%and $(w_x, p_x)$ and $(p_x, x)$ are created.
		\markj{Note that before these edges are added we already have $t_i(w)=t_i(p_x)$ for some parent $x$ of $p_x$.} %\todo[color=green!40]{Previously we said that edges $(w_x,p_x)$ and $(p_x,x)$ were added (without defining $w_x$). I think this fix is correct.}
		      From step 5 it follows that $t_i(v_y) =\tau+1>t_i(p_y)$, that $p(y)=\tau +2 > \tau +1$ and that $t(p_x)\leq \tau+1=t(v_y)$.
		      Therefore all of the  created edges are semi-temporal.
		      We also have \markj{$t_i(w)=t_i(p_x) < t_i(x)$, so these edges remain semi-temporal.}  %$t_i(w_x)<t_i(x)=t_i(p_x)$
		      %and $t_i(p_x)=t_i(x)$.
	\end{itemize}

	Note that the only place where non-temporal reticulation edges can be introduced is in \cref{item:case_edge_one} in the creation of edges $(v_y,p_x)$ for all $y\in Y$.
	This only happens for non-temporal items $s'_i$ and for each of this item at most $|Y|$ non-temporal reticulation edges are created, so $d(\N_i, t_i) \leq d(\N_{i+1}, t_{i+1}) + |Y|$.
	Because a non-temporal element $(x, Y)$ in $s'$ corresponds to $|Y|$ non-temporal elements in $s$, this implies that $d(\N,t) \leq d(s)$.
\end{proof}

%\begin{lem}
%	For a tree-child network $\N$ there exists a full tree-child sequence $s$ with $d(s)\leq d(\N)$ and $w_T(s)\leq r(\N)$. \todo{move proof to appendix, include intuitive explanation here}
%\end{lem}
\begin{lemSemiTemporalNetworkToSequence}
   \lemSemiTemporalNetworkToSequenceText
\end{lemSemiTemporalNetworkToSequence}
\begin{proof}
	We provide a way of constructing a tree-child sequence $s$ from a tree-child network $\N$ with semi-temporal labeling $t$ such that $d(s)=d(\N)$.
	We do this by modifying the proof from \cite[Lemma 3.4]{linz_attaching_2019}.
	Let $\rho$ denote the root of $\N$ and  let $v_1,\ldots, v_r$ denote the reticulations in the network.
	Let $\ell_\rho,\ell_1,\ldots, \ell_r$ denote the leaves at the end of the paths $P_\rho, P_1,\ldots, P_r$ starting at $v_1,\ldots, v_r$  respectively and consisting of only tree arcs.

	We will call a set $\{x, y\}$ with parents $p_x$ and $p_y$ in a given network a \emph{cherry} if $p_x=p_y$.
	We will call it a \emph{reticulated cherry} if $p_x$ and $p_y$ are joined by a reticulation edge $(p_y, p_x)$.
	In this case we call $x$ the \emph{reticulation leaf} of the cherry.
	We call such a reticulated cherry \emph{temporal} if $t(p_y)=t(p_x)$, otherwise we call it \emph{non-temporal}.

	Start off with an empty sequence $\sigma_0$.
	Set $\N_0:=\N$ and $i:=1$.

	\begin{enumerate}
		\item  If $\N_{i-1}$ consists of a single vertex $x$ then set $\sigma_i:=\sigma_{i-1} |((x,-))$ and return $\sigma_i$.
		\item If there is a cherry $\{x, y\}$ in $\N_{i-1}$, then
		      \begin{enumerate}
			      \item If one of $\{x,y\}$, say $x$, is an element of  $\{\ell_1, \ldots, \ell_r \}$ and $v_j$ is not a reticulation in $\N_{i-1}$ set $x_i:=x$ and $y_i:=y$.
			      \item Otherwise let $\{x_i, y_i\}:=\{x,y\}$ such that $x_i\notin \{\ell_p, \ell_1,\ldots, \ell_r \}$.
			      \item Set $\sigma_i=((x_i,y_i)) | \sigma_{i-1}$.
			            Let $\N_i$ be the tree-child network obtained from $\N_{i-1}$ by deleting $x_i$.
			      \item Go to step 5.
		      \end{enumerate}
		\item Else, if there is a non-temporal \markj{reticulated} cherry $\{x, y\}$ in $\N_{i-1}$ with $x$ the reticulation leaf then set $\sigma_i= \sigma_{i-1}|((x_i,y_i))$.
		      Let $\N_i$ be the tree-child network obtained from $\N_{i-1}$ by deleting the edge $(y_i,x_i)$ and suppressing vertices of both in-degree and out-degree one.
		\item Else, there has to be a temporal \markj{reticulated}  cherry $\{x, y\}$ in $\N_{i-1}$ with $x$ the reticulation leaf.
		      Let $q_1,\ldots,q_t$ be the set of leaves that $x$ is in a reticulation cherry with in $\N_{i-1}$.
		      Set $\sigma_i=\sigma_{i-1} | ((x,q_1),\ldots, (x,q_t))$.
		      Let $\N_i$ be the tree-child network obtained from $\N_{i-1}$ by deleting vertex $x$ and suppressing vertices of both in-degree and out-degree one.
		\item Increase $i$ and go to step 1.
	\end{enumerate}
	The proof that this yields a full tree-child sequence $s$ for $\N$ with $w_T(s)\leq r(\N)$ can be found in \cite[Lemma 3.4]{linz_attaching_2019}, so we will omit it here.
	Note that non-temporal elements can only be added to $s$ in step $3$ and each time this happens a non-temporal arc is removed from the network.
	Consequently the resulting tree-child sequence can not contain more non-temporal elements than the number of non-temporal arcs in $\N$.
	It follows that $d(s)\leq d(\N)$.
\end{proof}

%Second appearance
\begin{lemTCProcedureReturnsSequence}
   \lemTCProcedureReturnsSequenceText
\end{lemTCProcedureReturnsSequence}

%\begin{lem}
%	Let $s^\star$ be a tree-child sequence prefix, $T^\star$ a set of trees with the same leaves and define $T:=T^\star(s)$.
%	Suppose $k,p\in \mathbf{N}$ and $C\in \mathcal{L}(T)^2$.
%	When a generalized cherry picking sequence $s$ exists that satisfies $C$ and such that $s^ \star|s$ is a tree-child sequence for $T^\star$ with $w_{T^\star}(s^\star|s)\leq k^\star$ and $d(s)\leq p$ exists, \texttt{SemiTemporalCherryPicking}$(T, k, k^\star, p, C)$ from  \Cref{alg:better_temporal2} returns a non-empty set. \todo{Try to replace by intuitive explation and move proof to appendix.}
%	\label{lem:tc:procedure_returns_sequence}
%\end{lem}
\begin{proof}
	Let $W(k,u)$ be the claim that if a tree-child sequence $s$ for $T$ of weight at most $k$ exists that satisfies constraint set $C$ with $n^2-|C|\leq u$ and $d(s)\leq p$, such that $s^\star|s$ is a tree-child sequence of weight at most $k^\star$,  then the algorithm will return a non-empty set.
	We will prove this claim with induction on $k$ and $n^2-|C|$.

	For the base case $k=0$, if a generalized cherry picking sequence of weight $k$ exists we must have that all cherries in $T$ are trivial cherries.
	Therefore $|\Leav(T')|=1$, and a non-empty set is returned.

	Note that we can never have a constraint set $C$ with \markj{$|C|>n^2$} 
	%$|C|>n$
	because $C\subseteq \mathcal{L}(T)^2$.
	Therefore $W(k,-1)$ is true for all $k$.

	Now suppose $W(k, n^2-|C|)$ is true for all cases where $0\leq k< k_b$ and all cases where $k=k_b$ and $n^2-|C| \leq u$.
	We consider the case where a sequence $s$ with $d(s)\leq p$ of weight at most $k=k_b+1$ exists for $T$ that satisfies $C$ and $n^2-|C|\leq u_b+1$ such that $s^\star|s$ is a tree-child sequence for $T^\star$ with $w_{T^\star}(s^\star|s)\leq k^\star$.
	Now we will prove that a non-empty set is returned by the algorithm.
	.

	\cref{lem:bound_w_pc_nontemp} implies that $k-P(C)\geq 0$, so the if-statement on \cref{line2:first_if} will not be satisfied.
	From \cref{lem:tc:correctness_pick} it follows that a tree-child cherry picking sequence $s'$ for $T'$ of weight at most $k'$ exists for $T'$ that satisfies $C'$.
	From the way the \texttt{Pick} works it follows that either $k'<k$ or $n^2-C'= n^2-C$.
	If $|\Leav(T')|=1$ then $\{()\}$ is returned and we have proven $W(k_b+1, u+1)$ to be true for this case.
	Otherwise $s'$ is not empty, so $k'-P(C')\geq w_T(s')>0$.
	Because $s'$ satisfies $C'$, $\pi_1(C')\subseteq \Leav(T')$.
	So the condition on \cref{line2:third_if} is not satisfied.

	Now we are left with three cases:
	\begin{enumerate}
		\item If there is a pair $ (x, y) \in G(T',C')$ with  $w_T(x)>0 \land x\in \pi_1(C')$, then from \cref{lem:tc:branch_in_two} it follows that $s$ satisfies either $C'\cup \{(x,y)\}$ or $C'\cup \{(y,x)\}$.
		      From our induction hypothesis it now follows that either \texttt{SemiTemporalCherryPicking}($T'$, $k'$, $p$, $C'\cup \{ (x,y ) \}$) or  SemiTemporalCherryPicking($T'$, $k'$, $p$, $C'\cup \{ (y,x ) \}$) will return a non-empty set.
		      Therefore $R$ will not be empty, so a non-empty set will be returned.
		\item Otherwise, if there is a pair $ (x,a) \in G(T',C')$ with $w_{T'}(x)>0 \land x\notin \pi_2(C')$, there is a $b\neq a$ with $(x,b)\in G(T',C')$, for the same reasons as in 	\cref{lem:exists_uncovered_pair} for the temporal case.
		      Now from \cref{lem:tc:branch_in_three} it follows that $s$ satisfies $C'\cup \{(a,x)\}$, $C'\cup \{(b,x)\}$ or $C'\cup \{(x,a),(x,b)\}$.
		      From our induction hypothesis it now follows that the corresponding subcall will return a non-empty set.
		      Therefore $R$ will not be empty, so a non-empty set will be returned.
		\item Because the conditions in both the if and the else-if statement are not satisfied it follows that %$H(T)\setminus \pi_1(C)$
		\markj{$H(T)\setminus \pi_2(C)$}
		is empty. %\todo[color=green!40]{Added more argument for why $H(T)\setminus \pi_2(C)$ is empty, also changed $T$ to $T'$ in a few places}
		\markj{Indeed, any $x \in H(T)$ must have $w_{T'}(x) > 0$, as otherwise it would have been removed by \texttt{Pick}, and for any $x \in H(T)\setminus\pi_2(C)$ there exists at least one $y$ with $(x,y) \in G(T',C)$, as otherwise $\{(x,n) : n \in N_{T'}(x) \} \subseteq C'$ and again $x$ would be removed by \texttt{Pick}. Then either case 1 (if $x\in \pi_1(C)$) or case 2 (otherwise) would apply.}
		
		      Now from \cref{lem:first_non_temporal} it follows that $s_1$ has to be a non-temporal element.
		      Observe that for $s_1'=(x,y)$ we must have $x\notin \pi_2(C)$, because otherwise $s'$ has to contains some element $(z,x)$, but such an element can not appear after an element $(x,y)$, because the sequence is a tree-child sequence.
		      Also $y$ has to be in all trees in \markj{$T'$}, because otherwise $s$ contains an element $(y,z)$, which contradicts the assumption that $s|s'$ is tree-child.
		      So now we have shown that $s_1'\in P$.
		      Each element of $P$ is a cherry in \markj{$T'$}.
		      \Cref{lem:tc:unique_cherries} implies that there are at most $4k^\star$ unique cherries in \markj{$T'$}.
		      Therefore it follows that $|P|\leq 8k^\star$.
		      Because $s_1'\in P$, there $R$ is not empty because the result of SemiTemporalCherryPicking($T'(s_1'), k'-1,p-1,C''$) is added to $R$, which by our induction hypothesis is a non-empty set.
	\end{enumerate}
\end{proof}

\end{document}